\documentclass[conference]{IEEEtran}
\IEEEoverridecommandlockouts
\usepackage{amsmath,amssymb,amsfonts}
\usepackage{algorithmic}
\usepackage{graphicx}
\usepackage{textcomp}
\usepackage{xcolor}
\def\BibTeX{{\rm B\kern-.05em{\sc i\kern-.025em b}\kern-.08em
    T\kern-.1667em\lower.7ex\hbox{E}\kern-.125emX}}

\usepackage{hyperref}
\usepackage{cleveref}
\usepackage[numbers]{natbib}
\usepackage{multicol}
\usepackage{booktabs}
\usepackage{math_commands}
\usepackage{multirow}

\usepackage{xspace}
\usepackage{cleveref}
\usepackage{xcolor}
\usepackage{multirow}
\usepackage{tikz}
\usepackage{enumitem}
\usepackage{tablefootnote}
\usepackage{siunitx}

\newcommand\sysname{RANUM\xspace}
\newcommand\oldsysname{DEBAR\xspace}
\newcommand\NaN{\texttt{NaN}\xspace}
\newcommand\INF{\texttt{INF}\xspace}

\newcommand\fin{f^{\mathsf{in}}}
\newcommand\fout{f^{\mathsf{out}}}

\newcommand\gXvalid{\gX_{\mathsf{valid}}}
\newcommand\gWvalid{\gW_{\mathsf{valid}}}

\newcommand\tsys{t_{\mathsf{sys}}}
\newcommand\tunit{t_{\mathsf{unit}}}

\newcommand\vxtrain{\vx_{\mathsf{train}}}
\newcommand\vxinfer{\vx_{\mathsf{infer}}}

\newcommand\vwinfer{\vw_{\mathsf{infer}}}
\newcommand\vdwtarg{\Delta\vw_{\mathsf{targ}}}

\newcommand\valid{\mathsf{valid}}

\newcommand\precondgen{\mathsf{precond}}

\newcommand\Vfix{\gV_{\mathsf{fix}}}

\newcommand\sgn{\mathrm{sgn}}

\newcommand\minstep{\mathsf{minstep}}
\newcommand\maxiter{\mathsf{maxiter}}

\newcommand{\CodeIn}[1]{{\small\texttt{#1}}}
\newcommand{\SmallCodeIn}[1]{{\scriptsize\texttt{#1}}}

\newcount\Comments  
\usepackage{color}
\definecolor{darkgreen}{rgb}{0,0.5,0}
\definecolor{darkblue}{rgb}{0,0,0.5}
\definecolor{purple}{rgb}{1,0,1}
\definecolor{gray}{rgb}{0.5,0.5,0.5}
\definecolor{darkred}{RGB}{200,58,58}
\newcommand{\kibitz}[2]{\ifnum\Comments=0{\color{#1}{#2}}\fi}

\ifnum\Comments=0\else\excludeversion{old}\fi
\ifnum\Comments=0\newenvironment{new}{\color{blue}}{}\else\fi

\usepackage{versions}
\excludeversion{discard}

\usepackage{nicefrac}
\usepackage{enumitem}
\crefname{fact}{Fact}{Facts}

\usepackage{algorithm}
\usepackage{algorithmic}
\renewcommand\paragraph[1]{\textbf{#1}.\,}
\usepackage{balance}

\newcommand\circled[1]{\raisebox{0pt}{\textcircled{\raisebox{-0.9pt}{#1}}}}

\Crefname{subappendix}{Suppl.}{Suppl.}

\definecolor{revcolor}{HTML}{0a46f4}  

\newcount\arxiv 
\begin{document}




\title{Reliability Assurance for Deep Neural Network Architectures Against Numerical Defects}



\author{
    \IEEEauthorblockN{Linyi Li\IEEEauthorrefmark{1} \quad
    Yuhao Zhang\IEEEauthorrefmark{2} \quad
    Luyao Ren\IEEEauthorrefmark{3}\IEEEauthorrefmark{4} \quad
    Yingfei Xiong\IEEEauthorrefmark{3}\IEEEauthorrefmark{4} \quad
    Tao Xie\IEEEauthorrefmark{3}\IEEEauthorrefmark{4}\IEEEauthorrefmark{5}\thanks{\IEEEauthorrefmark{5}Corresponding author.}}
    \IEEEauthorblockA{\IEEEauthorrefmark{1}Department of Computer Science, University of Illinois Urbana-Champaign, 
    \href{mailto:linyi2@illinois.edu}{\textit{linyi2@illinois.edu}}}
    \IEEEauthorblockA{\IEEEauthorrefmark{2}Department of Computer Sciences, University of Wisconsin-Madison, 
    \href{mailto:yuhao.zhang@wisc.edu}{\textit{yuhao.zhang@wisc.edu}}}
    \IEEEauthorblockA{\IEEEauthorrefmark{3}School of Computer Science, Peking University, \href{mailto:rly@pku.edu.cn}{\textit{rly@pku.edu.cn}},
    \href{mailto:xiongyf@pku.edu.cn}{\textit{xiongyf@pku.edu.cn}},
    \href{mailto:taoxie@pku.edu.cn}{\textit{taoxie@pku.edu.cn}}}
    \IEEEauthorblockA{\IEEEauthorrefmark{4}Key Laboratory of High Confidence Software Technologies, Ministry of Education (Peking University)}
}


\maketitle

\ifnum\arxiv=1
\thispagestyle{plain}
\pagestyle{plain}
\fi


\begin{abstract}
With the widespread deployment of deep neural networks (DNNs), ensuring the reliability of DNN-based systems is of great importance. Serious reliability issues such as system failures can be caused by numerical defects, one of the most frequent defects in DNNs. 
To assure high reliability against numerical defects, 
in this paper, we propose the \sysname approach including novel techniques for three reliability assurance tasks: detection of potential numerical defects,  confirmation of potential-defect feasibility, and suggestion of defect fixes.
To the best of our knowledge, \sysname is the first approach that confirms potential-defect feasibility with failure-exhibiting tests and suggests fixes automatically.
Extensive experiments on the benchmarks of 63 real-world DNN architectures show that \sysname outperforms state-of-the-art approaches  across the three reliability assurance tasks. 
In addition, when the \sysname{}-generated fixes are compared with developers' fixes on open-source projects, in 37 out of 40 cases, \sysname{}-generated  fixes are equivalent to or even better than human fixes.  
\end{abstract}


%

\section{Introduction}
    \label{sec:intro}
    Deep Neural Networks~(DNNs) are successfully deployed and show remarkable performance in many challenging applications, including facial recognition~\cite{liu2015deep,zhang2016joint}, game playing~\cite{silver2017mastering}, and code completion~\cite{liu2022unified,bui2021infercode}. To develop and deploy DNNs, one needs to attain a DNN architecture, which is usually encoded by program code as the example shown in \Cref{fig:dnn-code-example}.
    First, for training, the user executes the program with the architecture on the given  training/validation data, attains the model weights, and stores them in a weight file.
    The architecture along with the weights is named a model.
    Then, for inference, the user loads the weight file to CPU/GPU memory or AI chips, executes the same program with the given inference sample and weights as arguments, and gets the model prediction result as the program output.
    With the wide deployment of DNN models (resulted from training DNN architectures), reliability issues of DNN-based systems have become a serious concern, where malfunctioning DNN-based systems have led to serious consequences such as fatal traffic accidents~\cite{drivingfatal}.
    
    To assure the reliability of DNN-based systems, it is highly critical to detect and fix numerical defects
    for two main reasons. 
    First, numerical defects widely exist in DNN-based systems. 
    For example, in the DeepStability database~\cite{kloberdanz2022deepstability}, over 250 defects are identified in deep learning~(DL) algorithms where over 60\% of them are numerical defects. 
    Moreover, since numerical defects exist at the architecture level, any model using the architecture naturally inherits these defects. 
    Second, numerical defects can result in serious consequences. Once numerical defects~(such as divide-by-zero) are exposed, the faulty DNN model will output \NaN or \INF instead of producing any meaningful prediction, resulting in numerical failures and system crashes~\cite{zhang2018empirical,zhang2019empirical}. Thus, numerical defects  hinder the application of DNNs in scenarios with high reliability and availability requirements such as threat monitoring in cybersecurity~\cite{problemaisecurity} and cloud system  controlling~\cite{sharma2016machine,jay2019deep}.
    
    \begin{figure}[t]
        \includegraphics[width=\linewidth]{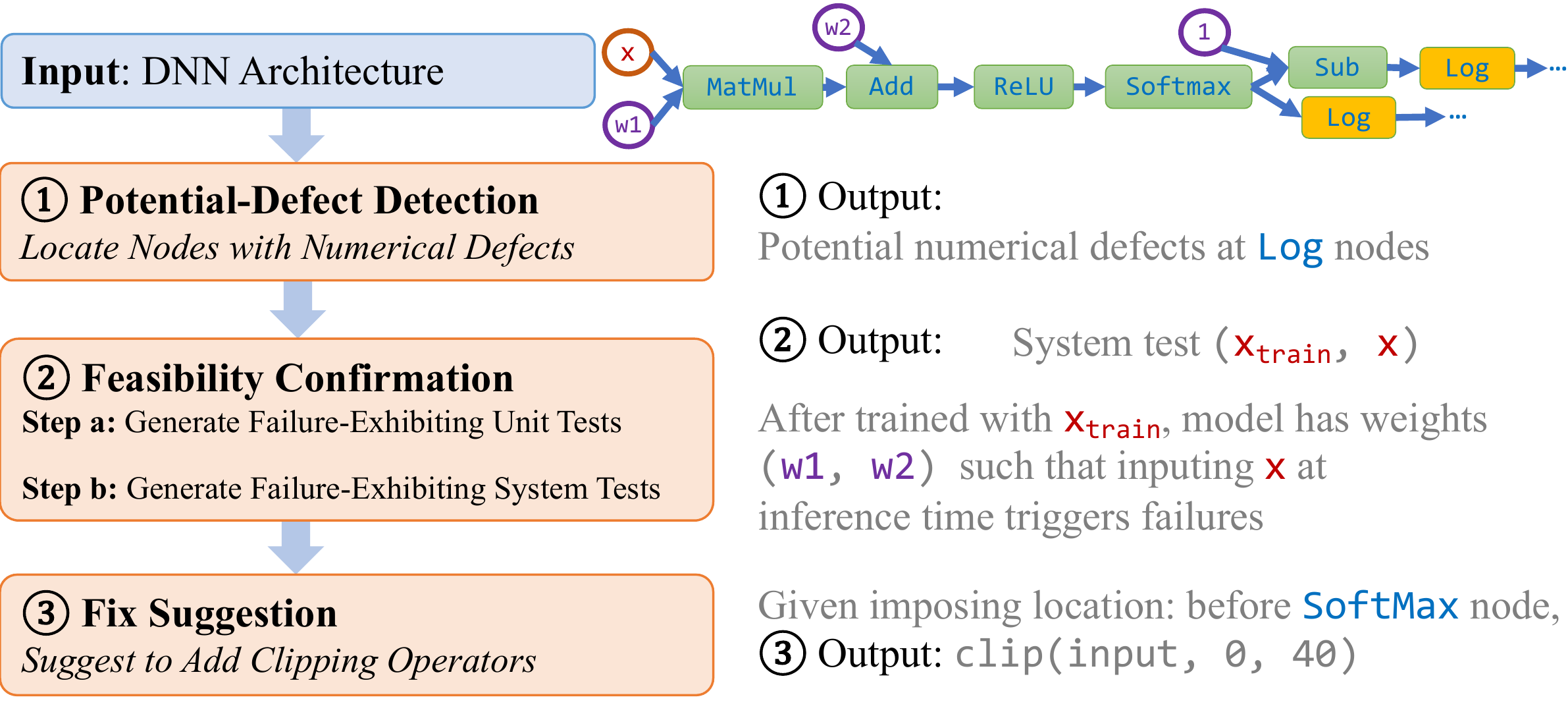}
        \vspace{-2.25em}
        \caption{Workflow for reliability assurance against numerical defects in DNN architectures.
        The left-hand side shows three tasks and the right-hand side shows corresponding examples.
        \sysname supports all the three tasks, and is the first automatic approach for system test generation and fix suggestion.}
        \label{fig:workflow}
        \vspace{-0.75em}
    \end{figure}
    
    To address numerical defects in DNN architectures in an actionable manner~\cite{SCIS22}, in this paper, we propose a workflow of reliability assurance (as shown in \Cref{fig:workflow}), consisting of three tasks: 
    potential-defect detection, feasibility confirmation, and fix suggestion, along with our proposed approach to support all these three tasks.

    \emph{Potential-Defect Detection.}
    In this task, we detect all potential numerical defects in a DNN architecture, with a focus on  operators with numerical defects (in short as defective operators) that potentially exhibit inference-phase numerical failures for two main reasons, 
    following the literature~\cite{zhang2020detecting,yan2021exposing}. First, these defective operators can be exposed after the model is deployed and thus are more devastating than those that potentially exhibit training-phase numerical failures~\cite{odena2019tensorfuzz,zhang2020detecting}. Second, a defective operator  that potentially exhibits training-phase numerical failures can usually be triggered to exhibit inference-phase numerical failures, thus also being detected by our task. 
    For example, the type of training-phase \NaN gradient failures is caused by an operator's input that leads to invalid derivatives, and this input also triggers failures in the inference phase~\cite{yan2021exposing}.
    
    \emph{Feasibility Confirmation}.
    In this task, we confirm the feasibility of these potential numerical defects by generating failure-exhibiting system tests.
    As shown in \Cref{fig:workflow}, a system test is a tuple of training example\footnote{In real settings, multiple training examples are used to train an architecture, but generating a single training example to exhibit failures (targeted by our work) is desirable for ease of debugging while being more challenging than generating multiple training examples to exhibit failures.} $\vxtrain$ and inference example $\vx$ such that after the training example is used to train the architecture under consideration, applying the resulting model on the inference example exhibits a numerical failure.

    \emph{Fix Suggestion.}
     In this task, we fix a feasible numerical defect. 
    To determine the fix form,
    we have inspected the developers' fixes of the numerical defects collected by \citeauthor{zhang2018empirical}~\cite{zhang2018empirical} by looking at follow-up Stack Overflow posts or GitHub commits.
    Among the 13 numerical defects whose fixes can be located,
    12 fixes can be viewed as explicitly or implicitly imposing interval preconditions on different locations, such as after inputs or weights are loaded and before defective operators are invoked.
    Thus, imposing an interval precondition, e.g., by clipping (i.e., chopping off the input parts that exceed the specified input range) the input for defective operator(s), is an effective and common strategy for fixing a numerical defect.
    Given a location~(i.e., one related to an  operator, input, or weight where users prefer to impose a fix), we suggest a fix for the  numerical defect under consideration.

    To support all the three tasks of the \underline{\textbf{r}}eliability \underline{\textbf{a}}ssurance process against DNN \underline{\textbf{num}}erical defects, we propose the \textbf{\sysname} approach in this paper.
    
    \emph{For task \circled{1} and task \circled{2}a, which are already supported by two existing tools (DEBAR~\cite{zhang2020detecting} and GRIST~\cite{yan2021exposing}), \sysname introduces novel extensions and optimizations that substantially improve the effectiveness and efficiency}. 
    (1)~DEBAR~\cite{zhang2020detecting} is the state-of-the-art tool for potential-defect detection; 
    however, DEBAR can handle only static computational graphs and does not support widely used dynamic graphs in PyTorch programs~\cite{paszke2019pytorch}.
    \sysname supports dynamic graphs thanks to our novel technique of \underline{\emph{backward fine-grained node labeling}}.
    (2)~GRIST~\cite{yan2021exposing} is the state-of-the-art  tool for generating failure-exhibiting unit tests to confirm potential-defect feasibility; however, GRIST conducts  gradient back-propagation by using the original inference input and weights as the starting point.
    Recent studies~\cite{madry2018towards,goodfellow2015explaining} on DNN adversarial attacks suggest that using a randomized input as the starting point leads to stronger attacks than using the original input.
    Taking this observation, we combine gradient back-propagation with random initialization in \sysname.
    
    \emph{For task \circled{2} and task \circled{3}, which are not supported by any existing tool, \sysname is the first automatic approach for them}. 
    
    For \textbf{feasibility confirmation}, \sysname is  the \textbf{first} approach that generates failure-exhibiting \textbf{system} tests that contain training examples. Doing so is a major step further from the existing GRIST tool, which generates failure-exhibiting unit tests ignoring the practicality of generated model weights.
    Given that in practice model weights are determined by training examples, we propose the technique of \underline{\emph{two-step generation}} for this task. 
    First, we generate a failure-exhibiting unit test. 
    Second, we generate a training example that leads to the model weights in the unit test when used for training.
    For the second step, we extend the deep-leakage-from-gradient~(DLG) attack~\cite{zhu2020deep} by incorporating the straight-through gradient estimator~\cite{bengio2013estimating}. 
    
    For \textbf{fix suggestion}, \sysname is the \textbf{first} automatic approach.
    \sysname is based on the novel technique of \underline{\emph{abstraction}} \underline{\emph{optimization}}. 
    We observe that a defect fix in practice is typically imposing interval clipping on some operators such that each later-executed operator (including those defective ones) can never exhibit numerical failures.
    Therefore, 
    we propose the novel technique of abstraction optimization to ``deviate away'' the input range of a defective operator from the invalid range, falling in which can cause numerical failures. 

For \sysname, we implement a tool\footnote{Open source at \texttt{\url{https://github.com/llylly/RANUM}}.} and evaluate it on the   benchmarks~\cite{yan2021exposing} of 63 real-world DNN architectures containing 79 real numerical defects; these benchmarks are the largest benchmarks of DNN numerical defects to the best of our knowledge.
    The evaluation results show that \sysname is both effective and efficient in all the three tasks for DNN reliability assurance.  
(1)~For potential-defect detection, \sysname detects $>$60\% more true defects than the state-of-the-art DEBAR approach. 
(2)~For feasibility confirmation, \sysname generates  failure-exhibiting unit tests
to confirm  potential numerical defects in the benchmarks with 100\% success rate; in contrast, with the much higher time cost (17.32X), the state-of-the-art GRIST approach generates unit tests to confirm defects with 96.96\% success rate. More importantly, for the first time, \sysname generates failure-exhibiting system tests that confirm defects (with 92.78\% success rate). 
(3)~For fix suggestion, \sysname proposes fix suggestions for  numerical defects with 100\% success rate. 
In addition, when the \sysname{}-generated fixes are compared with developers' fixes on open-source projects, in 37 out of 40 cases, \sysname{}-generated  fixes are equivalent to or even better than human fixes.

    This paper makes the following main contributions: 
    \begin{itemize}[leftmargin=*]
        \item 
        We formulate the reliability assurance problem for DNN architectures against numerical defects and elaborate on three important tasks for this problem.
    
        \item 
        We propose \sysname---the first automatic approach that solves all these three tasks.
        \sysname includes three novel techniques~(backward fine-grained node labeling, two-step test generation, and abstraction optimization) and solves system test generation and fix suggestion for the first time.
        
        \item We implement  \sysname and apply it on 63 real-world DNN architectures, showing the high effectiveness and efficiency of \sysname  compared to both the state-of-the-art approaches and developers' fixes.
    \end{itemize}

    \vspace{-0.25em}
\section{Background and Approach Overview}
    \label{sec:background-task}

    In this section, we introduce the background of DNN numerical defects and failures, and then give an overview of the \sysname approach with a running example. 

    \subsection{Background}
        \label{subsec:background}
    
    \begin{figure}[t]
        \centering
        \includegraphics[width=0.95\linewidth]{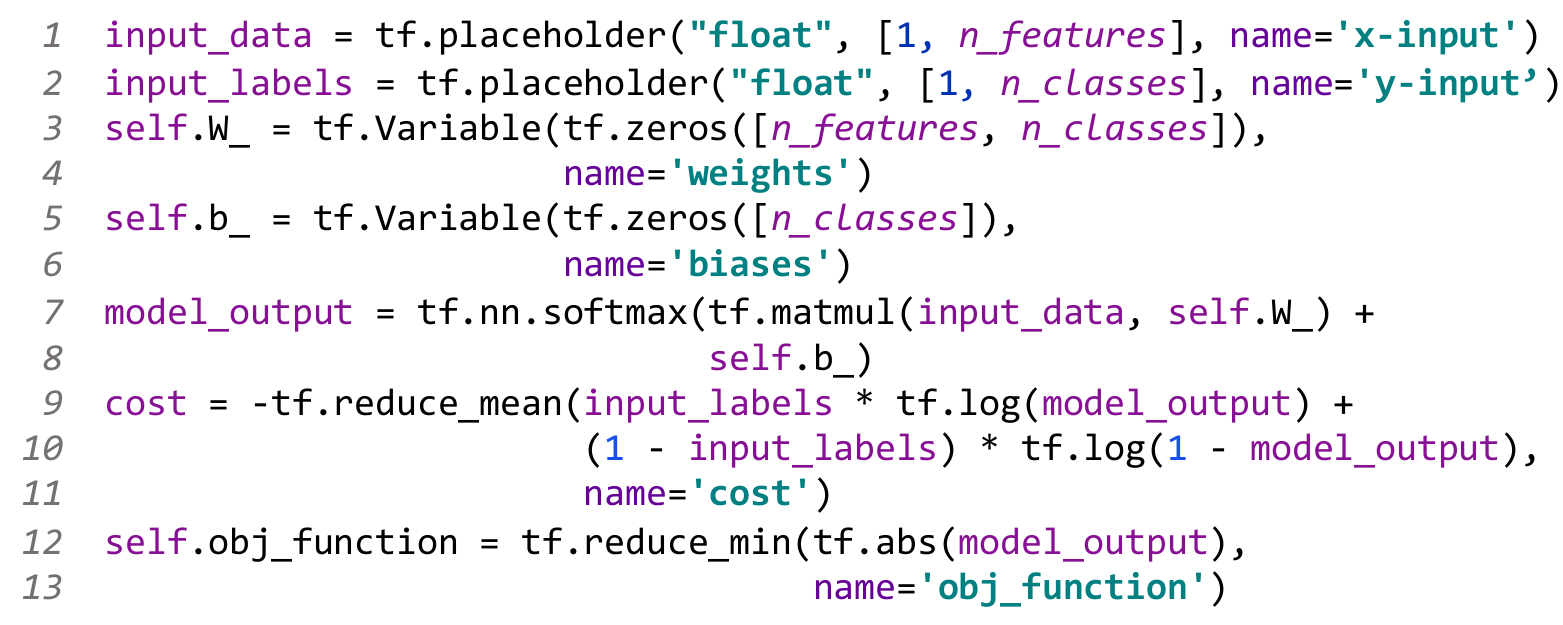}
        \vspace{-1.2em}
        \caption{A DL program snippet that defines a linear regression model from benchmarks of real-world numerical defects~(Case 2a in \cite{yan2021exposing}).}
        \label{fig:dnn-code-example}
        \vspace{-1.0em}
    \end{figure}

        DL developers define the DNN architecture with code using modern DL libraries such as PyTorch~\cite{paszke2019pytorch} and TensorFlow~\cite{tensorflow2015-whitepaper}.
        The DNN architecture can be expressed by a computational graph.
        \Cref{fig:dnn-code-example,fig:dnn-graph-example} depict a real-world example.
        Specifically, the DNN architecture  in a DL program can be automatically converted to an ONNX-format computational graph~\cite{onnx}.

        The computational graph can be viewed as a Directed Acyclic Graph~(DAG): $\gG = \langle \gV, \gE \rangle$, where $\gV$ and $\gE$ are sets of nodes and edges, respectively.
        We call nodes with zero in-degree as \emph{initial nodes}, which correspond to input, weight, or constant nodes.
        Initial nodes provide concrete data for the DNN models resulted from training the DNN architecture.
        The data from each node is formatted as a tensor, i.e., a multidimensional array, with a specified data type and array shape annotated alongside the node definition.
        We call nodes with positive in-degree as \emph{internal nodes}, which correspond to concrete operators, such as matrix multiplication~(\CodeIn{MatMul}) and addition~(\CodeIn{Add}).
        During model training, the model weights, i.e., data from weight nodes, are generated by the training algorithm.
        Then, in the deployment phase~(i.e., model inference), with these trained weights and a user-specified input named inference example,
        the output of each operator is computed in the topological order. 
        The output of some specific node is used as the prediction result.
        
        We let $\vx$ and $\vw$ denote the concatenation of data from all input nodes and data from all weight nodes,  respectively.\footnote{A bolded alphabet stands for a vector or tensor  throughout the paper.}
        For example, in \Cref{fig:dnn-graph-example}, $\vx$ concatenates data from nodes 1 and 11; and $\vw$ concatenates data from nodes 2 and 4.
        Given specific $\vx$ and $\vw$, the input and output for each node are deterministic.\footnote{An  architecture may contain stochastic nodes. We view these nodes as nodes with randomly sampled data, so the architecture itself is deterministic.}
        We use $\fin_n(\vx; \vw)$ and $\fout_n(\vx; \vw)$ to express input and output data of node $n$, respectively, given $\vx$ and $\vw$.
        
        \noindent
        \textbf{Numerical Defects in DNN Architecture.}
        We focus on inference-phase numerical defects.
        These defects lead to numerical failures when specific operators receive inputs within invalid ranges so that the operators output \NaN or \INF.

        \begin{definition}
            \label{def:numerical-defect}
            For the given computational graph  $\gG=\langle \gV, \gE\rangle$,
            if there is a node $n_0 \in \gV$, such that there exists a valid input and valid weights that can let the input of node $n_0$ fall within the invalid range, we say there is a \textbf{\emph{numerical defect}} at node~$n_0$.
            \vspace{-0.5em}
            \begin{equation*}
                \label{eq:numerical-defect}
                \begin{aligned}
                    \text{Formally,}\, & \exists \vx_0 \in \gXvalid, \vw_0 \in \gWvalid, 
                    \fin_{n_0}(\vx_0; \vw_0) \in \gI_{n_0,\mathsf{invalid}} \\
                    \Longrightarrow &
                    \text{$\exists$ numerical defect at node $n_0$}.
                \end{aligned}
            \end{equation*}
        \end{definition}
        In the definition, $\gXvalid$ and $\gWvalid$ are valid input range and weight range, respectively, which are clear given the deployed scenario.
        For example, ImageNet \CodeIn{Resnet50} models have valid input range $\gXvalid = [0,1]^{3\times 224\times 224}$ since image pixel intensities are within $[0,1]$, and valid weight range $\gWvalid = [-1,1]^{p}$ where $p$ is the number of parameters since weights of well-trained \CodeIn{Resnet50} models are typically within $[-1,1]$.
        The invalid range $\gI_{n_0,\mathsf{invalid}}$ is determined by $n_0$'s operator type with detailed definitions in \Cref{adxsec:list-numerical-defects}.
        For example, for the  \CodeIn{Log} operator, 
        the invalid range $\gI_{n_0,\mathsf{invalid}} = (-\infty, U_{\min})$ where $U_{\min}$ is the smallest positive number of a tensor's data type. 

        \vspace{-0.5em}
    \subsection{Approach Overview}
        \label{subsec:approach-overview}
        
        \begin{figure}[!t]
            \centering
            \includegraphics[width=0.9\linewidth]{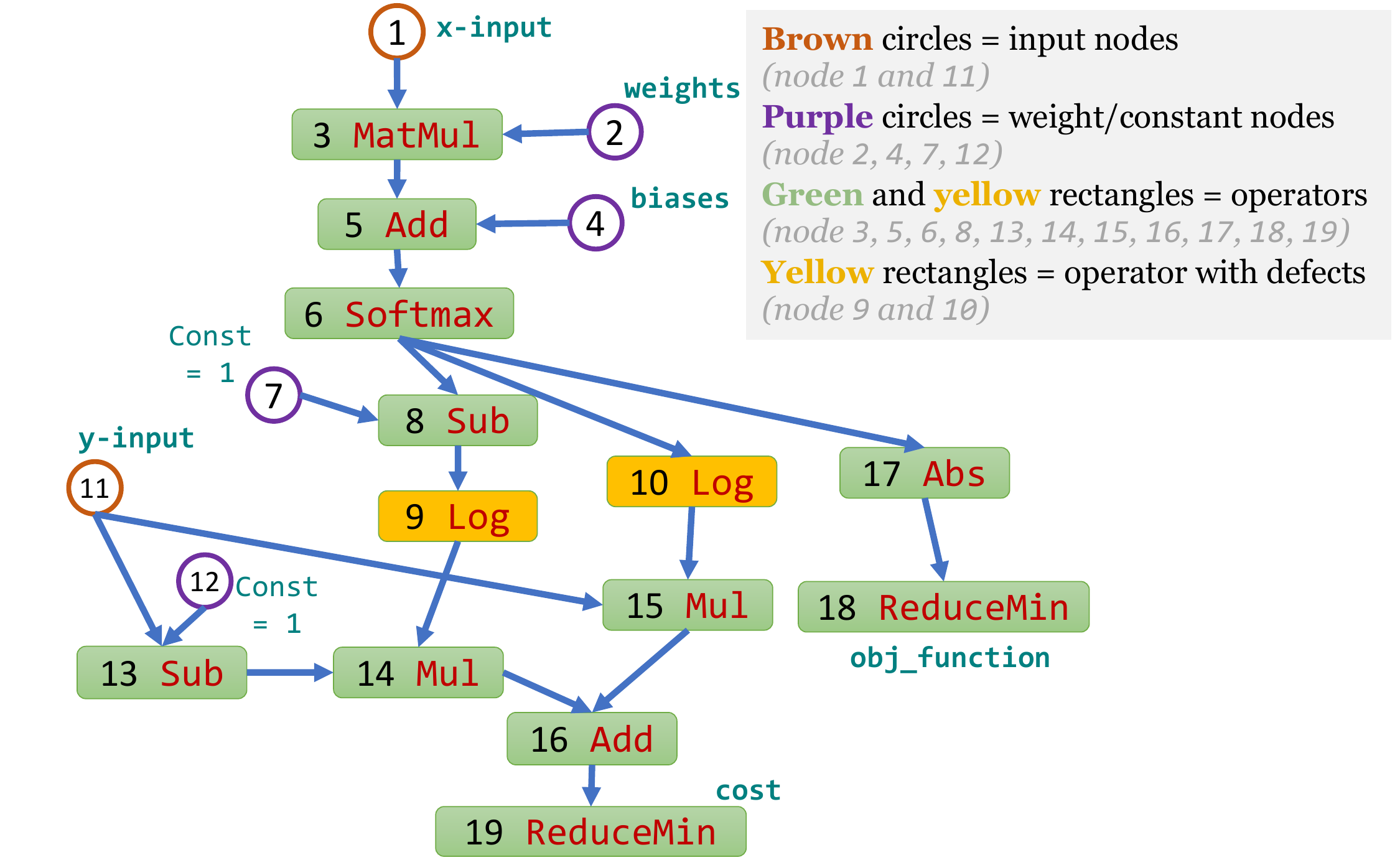}
            \vspace{-1em}
            \definecolor{figbrown}{RGB}{203,97,49}
            \definecolor{figpurple}{RGB}{112,56,153}
            \definecolor{figgreen}{RGB}{120,180,98}
            \definecolor{figyellow}{RGB}{235,163,27}
            \caption{Computational graph encoded by the snippet in \Cref{fig:dnn-code-example}.}

            \label{fig:dnn-graph-example}
            \vspace{-1.0em}
        \end{figure}
    
        \begin{figure}[t]
            \centering
            \includegraphics[width=0.90\linewidth]{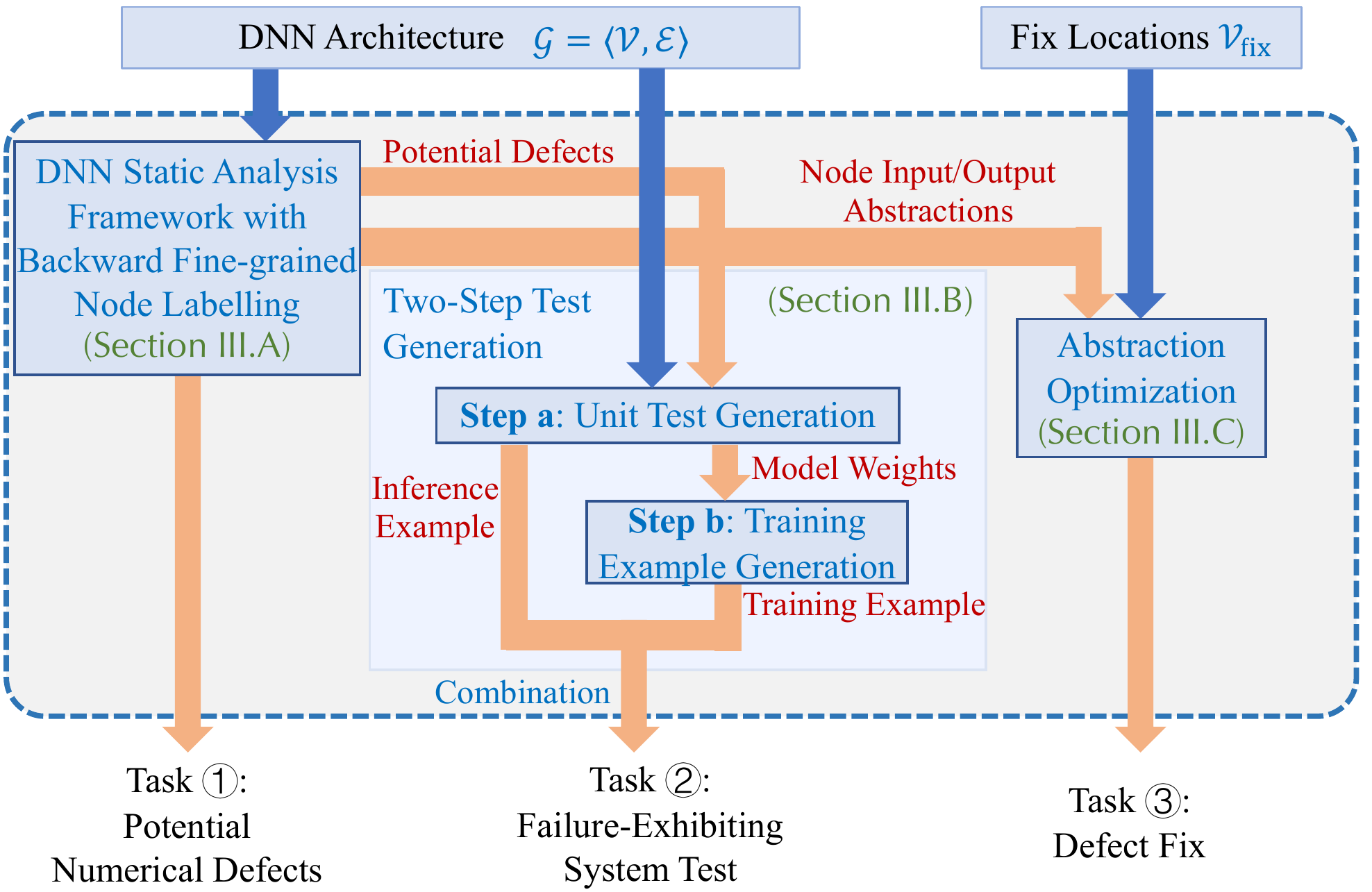}
            \vspace{-0.8em}
            \caption{Overview of the \sysname approach. 
            The output of \sysname indicates confirmation and manifestation of numerical defects (that can be feasibly exposed at the system level) for a given DNN architecture and effective fixes for the architecture's confirmed defects.
            }
            \label{fig:overview}
            \vspace{-1em}
        \end{figure}
        
        In \Cref{fig:overview}, we show the overview structure of the \sysname approach.
        \sysname takes a DNN architecture as the input.
        Note that although \sysname is mainly designed and illustrated for a DNN architecture, the \sysname approach can also be directly applied to general neural network architectures since they can also be expressed by computational graphs.
        First, the DNN static analysis framework (task \circled{1} in \Cref{fig:workflow}) in \sysname detects all potential numerical defects in the architecture.
        Second, the two-step test generation component~(task \circled{2} in \Cref{fig:workflow}), including unit test generation and training example generation,  confirms the feasibility of these potential numerical defects.
        Third, the abstraction optimization component (task \circled{3} in \Cref{fig:workflow}) takes the input/output abstractions produced by the DNN static analysis framework along with the user-specified fix locations, and produces preconditions to fix the confirmed defects. 
        
        We next go through the whole process in detail taking the DNN architecture shown in \Cref{fig:dnn-graph-example} as a running example.
        
        \paragraph{Task \circled{1}: Potential-Defect Detection via Static Analysis}
        The DNN static analysis framework within \sysname first computes the numerical intervals of possible inputs and outputs for all nodes within the given DNN architecture, and then flags any nodes whose input intervals overlap with their invalid ranges as nodes with potential numerical defects.
        
        In \Cref{fig:dnn-graph-example}, suppose that the user-specified input \CodeIn{x-input}~(node 1) is within (elementwise, same below) range $[(-10,-10)^\T, (10,10)^\T]$; \CodeIn{weights}~(node 2) are within range ${\scriptsize[\begin{bmatrix} -10 & -10 \\ -10 & -10 \end{bmatrix}, \begin{bmatrix} 10 & 10 \\ 10 & 10 \end{bmatrix}]}$; and \CodeIn{biases}~(node 4) are  within range $[(-10,-10)^\T, (10,10)^\T]$.
        Our DNN static analysis framework  computes these interval abstractions for node inputs:
        \begin{enumerate}[leftmargin=*]
            \item Node 5~(after \CodeIn{MatMul}): $[(-200,-200)^\T, (200,200)^\T]$;
            \item Node 6~(after \CodeIn{Add}): $[(-210,-210)^\T, (210,210)^\T]$;
            \item Node 8~(after \CodeIn{Softmax} in \CodeIn{float32}): $[(0,0)^\T, (1,1)^\T]$;
            \item Node 9~(after \CodeIn{Sub} of $[1,1]$ and node 8), 10: $[(0,0)^\T, (1,1)^\T]$.
        \end{enumerate}
        Since nodes 9 and 10 use the \CodeIn{Log} operator whose invalid input range $(-\infty, U_{\min})$ overlaps with their input range $[(0,0)^\T, (1,1)^\T]$, we flag nodes 9 and 10 as potential numerical defects.
        
        This static analysis process follows the state-of-the-art DEBAR tool~\cite{zhang2020detecting}.
        However, we extend DEBAR with a novel technique named backward fine-grained node labeling. 
        This technique detects all nodes that require fine-grained abstractions, e.g., nodes that determine the control flow in a dynamic graph. 
        For these nodes, we apply interval abstractions with the finest granularity to reduce control flow ambiguity.
        For other nodes, we let some neighboring elements share the same interval abstraction to improve efficiency while preserving tightness.
        As a result, the static analysis in \sysname has high efficiency and supports much more DNN operators including dynamic control-flow operators like \CodeIn{Loop} than DEBAR does.
        
        \paragraph{Task \circled{2}: Feasibility Confirmation via Two-Step Test Generation}
        Given nodes that contain potential numerical defects~(nodes 9 and 10 in our example), we generate failure-exhibiting system tests to confirm their feasibility. 
        A failure-exhibiting system test is a tuple $\langle\vxtrain, \vxinfer\rangle$, such that after training the architecture with the training example $\vxtrain$,\footnote{In particular, if our generation technique outputs $\vxtrain$, the numerical failure can be triggered if the training dataset contains \emph{only} $\vxtrain$ or \emph{only} multiple copies of $\vxtrain$ and the inference-time input is $\vxinfer$. Our technique can also be applied for generating a batch of training examples by packing the batch as a single example: $\vxtrain = ({\vxtrain}_1, {\vxtrain}_2, \dots, {\vxtrain}_B)$.} with the trained model weights $\vwinfer$, the inference input $\vxinfer$ triggers a numerical failure.
        The name ``system test'' is inspired by traditional software testing, where we test the method sequence~(\CodeIn{m = train($\vxtrain$); m.infer($\vxinfer$)}). 
        In contrast, GRIST~\cite{yan2021exposing} generates model weights $\vwinfer$ along with inference input $\vxinfer$ that tests only the inference method \CodeIn{m.infer()}, and the weights may be infeasible from training. 
        Hence, we view the GRIST-generated tuple $\langle\vwinfer, \vxinfer\rangle$ as a ``unit test''.

        We propose a two-step test generation technique to generate failure-exhibiting system tests.
        
        \emph{\textbf{Step a}: Generate failure-exhibiting unit test $\langle\vwinfer, \vxinfer\rangle$.}
        The state-of-the-art GRIST tool supports this step.
        However, GRIST solely relies on gradient back-propagation, which is relatively inefficient.
        In \sysname, we augment GRIST by combining its gradient back-propagation with random initialization inspired by recent research on DNN adversarial attacks~\cite{goodfellow2015explaining,madry2018towards}.
        As a result, \sysname achieves 17.32X speedup with 100\% success rate.
        Back to the running example in \Cref{fig:dnn-graph-example},  \sysname can generate ${\scriptsize\begin{bmatrix} 5 & -5 \\ -5 & 5 \end{bmatrix}}$ for node 2 and $(0.9, -0.9)^\T$ for node 4 as model weights $\vwinfer$; and $(10, -10)^\T$ for node 1 and $(1,0)^\T$ for node 11 as the inference input $\vxinfer$.
        Such $\vwinfer$ and  $\vxinfer$ induce input $(0,1)^\T$ and $(1,0)^\T$ for nodes 9 and 10, respectively.
        Since both nodes 9 and 10 use the \CodeIn{log} operator and $\log 0$ is undefined, both nodes 9 and 10 trigger numerical failures.
        
        \emph{\textbf{Step b}: Generate training example $\vxtrain$ that achieves model weights $\vwinfer$.}
        To the best of our knowledge, there is no automatic approach for this task yet.
        \sysname provides support for this task based on our extension of DLG attack~\cite{zhu2020deep}.
        The DLG attack is originally designed for recovering the training data from training-phase gradient leakage.
        Here, we figure out the required training gradients to trigger the numerical failure at the inference phase and then leverage the DLG attack to generate $\vxtrain$ that leads to such training gradients.
        Specifically, many DNN architectures contain operators (such as \CodeIn{ReLU}) on which DLG attack is hard to operate~\cite{serra2018bounding}.
        We combine straight-through estimator~\cite{bengio2013estimating} to provide proxy gradients and bypass this barrier.
        Back to the running example in \Cref{fig:dnn-graph-example}, supposing that the initial weights are ${\scriptsize\begin{bmatrix} -0.1 & 0.1 \\ 0.1 & -0.1 \end{bmatrix}}$ for node 2 and $(0,0)^\T$ for node 4,  \sysname can generate training example $\vxtrain$ composed of $(5.635, -5.635)^\T$ for node 1 and $(1,0)^\T$ for node 11, such that one-step training with learning rate 1 on this example leads to $\vwinfer$.
        Combining $\vxtrain$ from this step with $\vxinfer$ from \emph{step a}, we obtain a failure-exhibiting system test that confirms the feasibility of potential defects in nodes 9 and 10.
        
        
        \paragraph{Task \circled{3}: Fix Suggestion via Abstract Optimization}
        In this task, we suggest fixes for the confirmed numerical defects.
        \sysname is the first approach for this task to our knowledge.
        
        The user may prefer different fix locations, which correspond to a user-specified set of nodes $\Vfix \subseteq \gV$ to impose the fix.
        For example, if the fix method is clipping the inference input, $\Vfix$ are input nodes~(e.g., nodes 1, 11 in \Cref{fig:dnn-graph-example});
        if the fix method is clipping the model weights during training, $\Vfix$ are weight nodes~(e.g., nodes 2, 4 in \Cref{fig:dnn-graph-example});
        if the fix method is clipping before the defective operator,
        $\Vfix$ are nodes with numerical defects~(e.g., nodes 9, 10 in \Cref{fig:dnn-graph-example}).
            
        According to the empirical study of developers' fixes in \Cref{sec:intro}, 12 out of 13 defects are fixed by imposing interval preconditions for clipping the inputs of $\Vfix$.
        Hence, we suggest interval precondition, which is interval constraint $\vl_n \le \fin_{n}(\vx; \vw) \le \vu_n$ for nodes $n\in \Vfix$, as the defect fix in this paper.
        A fix should satisfy that, when these constraints $\bigwedge_{n\in\Vfix} (\vl_n \le \fin_{n}(\vx; \vw) \le \vu_n)$ are imposed, the input of any node in the computational graph should always be valid, i.e., $\fin_{n_0}(\vx; \vw) \notin \gI_{n_0, \mathsf{invalid}}, \forall n_0 \in \gV$.
        
        In \sysname, we formulate the fix suggestion task as a constrained optimization problem, taking the endpoints of interval abstractions for nodes in $\Vfix$ as optimizable variables.
        We then propose the novel technique of abstraction optimization to solve this constrained optimization problem.
        Back to the \Cref{fig:dnn-graph-example} example, if users plan to impose a fix on inference input, \sysname can suggest the fix $-1 \le $ \CodeIn{x-input} $\le 1$;
        if users plan to impose a fix on nodes with numerical defects, \sysname can suggest the fix $10^{-38} \le $ node 9 \& node 10\CodeIn{.input} $\le +\infty$.

\section{The \sysname Approach}

    \label{sec:method}
    
    In this section, we introduce the three novel techniques in \sysname: backward fine-grained node labeling in \Cref{subsec:static-analysis};
    two-step test generation in \Cref{subsec:system-test-gen}; 
    and abstraction optimization in \Cref{subsec:precond-gen}.

    \subsection{DNN Static Analysis Framework with Backward Fine-Grained Node Labeling for Potential-Defect Detection}
    
        \label{subsec:static-analysis}
        
        \sysname contains a static analysis framework to enable potential-defect detection and support downstream tasks as shown in \Cref{fig:overview}.
        Given a DNN architecture and valid ranges for input and weight nodes, the static analysis framework computes interval abstractions for possible inputs and outputs of each node. 
        As a result, we can check whether an overlap exists between the interval abstraction and invalid input ranges for all nodes in the graph to detect potential numerical defects.
        Then, the defective nodes are fed into the two-step test generation component to confirm the feasibility of potential defects; 
        and the differentiable abstractions are fed into the abstract optimization component to produce fixes.
        
        Formally, for given valid ranges of inference input and model weights, namely $\gX$ and $\gW$, for each node $n\in\gV$, our framework  computes \emph{sound} input interval abstraction $[\vl_n, \vu_n] := \{\vx: \vl_n \le \vx \le \vu_n\}$ such that $[\vl_n, \vu_n]$ always captures all possible inputs of the node: $[\vl_n, \vu_n] \supseteq \{\fin_n(\vx, \vw): \vx\in\gX, \vw\in\gW\}$.
        We also compute output interval abstractions similarly.
        
        Compared with traditional analysis tools for numerical software~\cite{gurfinkel2015seahorn,singh2017fast}, \sysname's static analysis framework designs abstractions for DNN primitives operating on multi-dimensional tensors that are not supported by traditional tools.
        Compared with the state-of-the-art DEBAR tool~\cite{zhang2020detecting}, \sysname uses the same abstraction domain~(interval domain with tensor partitioning), but incorporates a novel technique~(backward fine-grained node labeling) to improve abstraction precision and support a wider range of DNN architectures.

        \textbf{Abstract Domain: Interval with Tensor Partitioning.}
        Following \oldsysname's design, we use the interval with tensor partitioning~\cite{zhang2020detecting} as the abstraction domain. This abstraction domain partitions a  tensor into multiple  subblocks and shares the interval abstractions at the block level instead of imposing abstractions at the element level. 
        Therefore, we can compute the abstraction of a smaller size than the original tensor to improve efficiency.
        
        \paragraph{Our Technique: Backward Fine-Grained Node Labeling}
        The interval domain with tensor partitioning provides a degree of freedom in terms of the partition granularity, i.e., we can choose the subblock size for each node's abstraction.
        When the finest granularity, i.e., elementwise abstraction, is chosen, the abstraction interval is the most concrete.
        When the coarsest granularity~(i.e., one scalar to summarize the node tensor) is chosen, the abstraction saves the most space and computational cost but loses much precision.

        \emph{Example.}
        Suppose that the possible input range of a node is $([-1,0], [0,1], [1,2], [-1,0])$, where each interval $[l,u]$ specifies the range of corresponding elements in the four-dimensional vector.
        If we choose the finest granularity, we  use $[\vl_n,\vu_n] = [(-1,0,1,-1), (0,1,2,0)]$ as the input range abstraction.
        If we choose the coarsest granularity, we  use $[\vl_n, \vu_n] = [-1,2]$ as the abstraction where the same interval is shared for all elements.
        As we can see, finer granularity provides tighter abstraction at the expense of larger computational and space costs.
        
        In DEBAR, the coarsest granularity is used by default for most operators.
        However, we find that using the finest instead of the coarsest granularity for some nodes is more beneficial for overall abstraction preciseness.
        For example, the control-flow operators~(e.g., \CodeIn{Loop}) benefit from concrete execution to determine the exact control flow in the dynamic graph, and the indexing operators~(e.g., \CodeIn{Slice}) and shaping operators~(e.g., \CodeIn{Reshape}) benefit from explicit indexers and shapes to precisely infer the output range.
        Hence, we propose to use the finest granularity for some nodes~(namely fine-grained requiring operators) while the coarsest granularity for other nodes during static analysis.

        To benefit from the finest granularity abstraction for required nodes, typically, all of their preceding nodes also need the finest granularity.
        Otherwise, the over-approximated intervals from preceding nodes will be propagated to the required nodes, making the finest abstraction for the required nodes useless.
        To solve this problem, in \sysname, we back-propagate ``fine-grained'' labels from these fine-grained requiring nodes to initial nodes by topologically sorting the graph with \emph{inverted} edges, and then apply the finest granularity abstractions on all labeled nodes.
        In practice, we find that this strategy eliminates the control-flow ambiguity and indexing ambiguity with little loss of efficiency\footnote{Theoretically, using the finest granularity for tensor partitioning cannot fully eliminate the ambiguity, since interval abstraction is intrinsically an over-approximation. Nevertheless, in our evaluation (Section~\ref{sec:exp}), we find that this technique eliminates control-flow and indexing ambiguities on all 63 programs in the benchmarks.}.
        As a result, \sysname supports all dynamic graphs~(which are not supported by DEBAR) that comprise 39.2\% of the benchmarks proposed by \citeauthor{yan2021exposing}~\cite{yan2021exposing}.
        
        Furthermore, when preceding nodes use finer-grain abstraction granularity, the subsequent nodes should preserve such fine granularity to preserve the analysis preciseness.
        Principally, the choice of abstraction granularity should satisfy both tightness~(bearing no precision loss compared to elementwise interval abstraction) and minimality~(using the minimum number of partitions for high efficiency).
        To realize these principles, we dynamically determine a  node's abstraction granularity based on the granularity of preceding nodes.
        The abstraction design for some operators is non-trivial.
        Omitted details~(formulation, illustration, and proofs) about the static analysis framework are in \Cref{adxsec:static-analysis}.

        In summary, the whole static analysis process consists of three steps.  
        (1)~Determine the tensor partition granularity of all initial nodes by our technique of backward fine-grained node labeling.
        (2)~Sort all nodes in the graph in the topological order.
        (3)~Apply corresponding abstraction computation algorithms for each node based on the preceding node's abstractions.
        The key insight behind the design of our static analysis framework is the strategic granularity selection for tensor abstraction, maintaining both high efficiency~(by selecting the coarse granularity for data-intensive nodes) and high precision~(by selecting the fine granularity for some critical nodes, such as nodes with control-flow, indexing, and shaping operators).
        
    \subsection{Two-Step Test Generation for Feasibility Confirmation}
    
        \label{subsec:system-test-gen}
    
        \sysname generates failure-exhibiting system tests for the given DNN to confirm the feasibility of potential numerical defects.
        Here, we take the DNN architecture as the input.
        From the static analysis framework, we obtain a list of nodes that have potential numerical defects.
        For each node $n_0$ within the list, we apply our technique of two-step test generation to produce a failure-exhibiting system test $\tsys = \langle \vxtrain, \vxinfer \rangle$ as the output.
        According to \Cref{subsec:approach-overview}, the test should satisfy that after the architecture is trained with $\vxtrain$, entering $\vxinfer$ in the inference phase results in a numerical failure.
        
        We propose the novel technique of two-step test generation:
        first, generate failure-exhibiting unit test $\langle \vwinfer, \vxinfer \rangle$;
        then, generate training example $\vxtrain$ that leads model weights to be close to $\vwinfer$ after training.
        
        \paragraph{\textbf{Step a}: Unit Test Generation}
            As sketched in \Cref{subsec:approach-overview}, we strengthen the state-of-the-art unit test generation approach, GRIST~\cite{yan2021exposing}, by combining it with random initialization to complete this step.
            Specifically, GRIST leverages the gradients of the defective node's input with respect to the inference input and weights to iteratively update the inference input and weights to generate failure-exhibiting unit tests.
            However, GRIST always conducts updates from the existing inference input and weights, suffering from local minima problem~\cite{madry2018towards}.
            Instead, motivated by DNN adversarial attack literature~\cite{madry2018towards,tramer2020adaptive}, a sufficient number of random starts help  find global minima effectively.
            Hence, in \sysname, we first conduct uniform sampling 100 times for both the inference input and weights to trigger the numerical failure.
            If no failure is triggered, we use the sample that induces the smallest loss as the start point for gradient optimization.
            As \Cref{subsec:static-detector-exp} shows, this strategy substantially boosts the efficiency, achieving 17.32X speedup.

        \paragraph{\textbf{Step b}: Training Example Generation}
            For this step, \sysname takes the following inputs: (1)~the DNN architecture,  (2)~the failure-exhibiting unit test $\tunit = \langle \vwinfer, \vxinfer \rangle$,  and (3) the randomly initialized weights $\vw_0$.
            Our goal is to generate a legal training example $\vxtrain$, such that the model trained with $\vxtrain$ will contain weights close to $\vwinfer$.

            DNNs are typically trained with gradient-descent-based algorithms such as stochastic gradient descent~(SGD).
            In SGD, in each step $t$, we sample a mini-batch of samples from the training dataset to compute their gradients on model weights and use these gradients  to update the weights.
            We focus on one-step SGD training with a single training example, since generating a single one-step training example to exhibit a failure is more desirable for  debugging because, in one-step training, the model weights are updated strictly following the direction of the gradients. 
            Therefore, developers can inspect inappropriate weights, easily trace back to nodes with inappropriate gradients, and then fix these nodes. 
            In contrast, in multi-step training, from inappropriate weights, developers cannot trace back to inappropriate gradients because weights are updated iteratively and interactions between gradients and weights are complex~(even theoretically intractable~\cite{li2017convergence}).
            
            In this one-step training case, after training, the model weights $\vwinfer$ satisfy 
            \vspace{-0.5em}
            \begin{equation}
                \vwinfer = \vw_0 - \gamma \nabla_\vw \gL(\vxtrain; \vw_0),
                \vspace{-0.5em}
                \label{eq:one-step-sgd}
            \end{equation}
            where $\gamma \in \sR_+$ is a predefined learning rate, and $\gL$ is the predefined loss function in the DNN architecture.
            Hence, our goal becomes finding $\vxtrain$ that satisfies 
            \vspace{-0.3em}
            \begin{equation}
                \nabla_\vw \gL(\vxtrain; \vw_0) = (\vw_0 - \vwinfer) / \gamma. \label{eq:dlg-attack}
            \end{equation}
            \vspace{-1.5em}

            The DLG attack~\cite{zhu2020deep} is a technique for generating input data that induce specific weight gradients.
            The attack is originally designed for recovering training samples from monitored gradient updates.
            Since the right-hand side~(RHS) of \Cref{eq:dlg-attack} is known, our goal here is also to generate input example $\vxtrain$ that induces specific weight gradients.
            Therefore, we leverage the DLG attack to generate training example $\vxtrain$.
            
            \paragraph{Extending DLG Attack with Straight-Through Estimator}
            Directly using DLG attack suffers from an optimization challenge in our scenario.
            Specifically, in DLG attack, suppose that the target weight gradients are $\vdwtarg$, we use gradient descent over the squared error $\| \nabla_\vw \gL(\vx; \vw_0) - \vdwtarg \|_2^2$ to generate $\vx$.
            In this process, we need meaningful gradient information of this squared error loss to perform the optimization.
            However, the gradient of this loss involves second-order derivatives of $\gL(\vx;\vw_0)$, which could be zero. 
            For example, DNNs with \CodeIn{ReLU} as activation function are piecewise linear and have zero second-order derivatives almost everywhere~\cite{serra2018bounding}.
            This optimization challenge is partly addressed in DLG attack by replacing \CodeIn{ReLU} with \CodeIn{Sigmoid}, but it changes the DNN architecture~(i.e., the system under test) and hence is unsuitable.
            
            We leverage the straight-through estimator to mitigate the optimization challenge.
            Specifically, for a certain  operator, such as \CodeIn{ReLU}, we do not change its forward computation but change its backward gradient computation to provide second-order derivatives within the DLG attack process.
            For example, for \CodeIn{ReLU}, in backward computation we use the gradient of $\mathrm{Softplus}$ function, namely $1 - \frac{1}{1+\exp(\vx)}$, 
            because $\mathrm{Softplus}$ is an approximation of \CodeIn{ReLU}~\cite{glorot2011deep} with non-zero second-order derivatives.
            Note that we modify the computed gradients only within the DLG attack process.
            After such $\vxtrain$ is generated by the attack, we evaluate whether it triggers a numerical failure using the original architecture and gradients in \Cref{eq:one-step-sgd}.

            \Cref{adxsec:detail-hyperparam} lists hyperparameters used by our implementation.
            
    \subsection{Abstraction Optimization for Fix Suggestion}
    
        \label{subsec:precond-gen}
        
        In this task, we aim to generate the precondition fix given imposing locations.
        The inputs are the DNN architecture, the node $n_0$ with numerical defects, and a node set $\Vfix$ to impose the fix.
        We would like to generate interval preconditions for $\Vfix$ node inputs so that after these preconditions are imposed, the defect on $n_0$ is fixed.
        
        Formally, our task is to find $\langle l_n, u_n \rangle$ for each $n\in \Vfix$~($l_n$ and $u_n$ are scalars so the same interval bound applied to all elements of $n$'s tensor), such that for any $\vx, \vw$ satisfying $\fin_{n}(\vx;\vw) \in [l_n, u_n]$, $\forall n\in \Vfix$, for the defective node $n_0$, we have $\fin_{n_0}(\vx;\vw) \notin \gI_{n_0,\textsf{invalid}}$, where the full list of invalid input ranges $\gI_{n_0, \textsf{invalid}}$ is in \Cref{adxsec:list-numerical-defects}.
        There is an infinite number of possible $\langle l_n, u_n \rangle$ interval candidates since $l_n$ and $u_n$ are floating numbers.
        Hence, we need an effective technique to find a valid solution from the exceedingly large search space that incurs a relatively small model utility loss.
        To achieve so, we formulate a surrogate optimization problem for this task.
        \vspace{-0.5em}
        \allowdisplaybreaks
        \begin{small}
        \begin{align}
            \maximize_{l_n, u_n: n\in \Vfix} \quad s \quad
            \mathrm{s.t.}\quad & u_n \ge l_n + s(u_n^{\valid} - l_n^{\valid}), \forall n\in \Vfix, \label{eq:precond-span-cons} \\
            & l_n^{\valid} \le l_n \le u_n \le u_n^{\valid}, \forall n\in \Vfix, \label{eq:precond-valid-cons} \\
            & \gL_{n_0}^{\precondgen}(\{l_n, u_n\}_{n\in\Vfix}) < 0. \label{eq:precond-effective-cons}
        \end{align}
        \end{small}
        Here, 
        $l_n^{\valid}$ and $u_n^{\valid}$ are the valid ranges~(of the node's input $n$), which are fixed and determined by the valid ranges of input and weights.
        $\gL_{n_0}^{\precondgen}$ is the node-specific precondition generation loss that is the distance between the furthest endpoint of defective node $n_0$'s interval abstraction and $n_0$'s valid input range.
        Hence, when $\gL_{n_0}^{\precondgen}(\{l_n, u_n\}_{n\in\Vfix})$ becomes negative, the solution $\{l_n, u_n\}_{n\in \Vfix}$ is a valid precondition.
        The optimization variables are the precondition interval endpoints $l_n$ and $u_n$ and the objective is the relative span of these intervals.
        The larger the span is, the looser the precondition constraints are, and the less hurt they are for the model's utility.
        \Cref{eq:precond-span-cons} enforces the interval span requirement.
        \Cref{eq:precond-valid-cons} assures that the precondition interval is in the valid range.
        \Cref{eq:precond-effective-cons} guarantees the validity of the precondition as a fix.
        
        For any $\{l_n, u_n\}_{n\in \Vfix}$, thanks to \sysname's static analysis framework, we can compute induced intervals of defective node $n_0$, and thus compute the loss value $\gL_{n_0}^{\precondgen}$.
        
        \setlength{\textfloatsep}{0pt}
        \begin{algorithm}[!t]
            \caption{Abstraction Optimization~(\Cref{subsec:precond-gen})}
            \label{alg:abs-opt}
            \begin{algorithmic}[1]
                \small
                \renewcommand{\algorithmicrequire}{\textbf{Input:}}
                \renewcommand{\algorithmicensure}{\textbf{Output:}}
                \REQUIRE DNN architecture $\gG=\langle \gV, \gE \rangle$, defective node $n_0 \in \gV$, nodes to impose fix $\Vfix \subseteq \gV$
                \STATE $s\gets 1, \gamma_s\gets 0.9, \gamma_c\gets 0.1, \minstep\gets 0.1, \maxiter\gets 1000$ 
                \STATE $c_n \gets (l_n^{\valid} + u_n^{\valid}) / 2, l_n \gets l_n^{\valid}, u_n \gets u_n^{\valid},\,  \forall n \in \Vfix$
                \FOR {$i = 1$ to $\maxiter$}
                    \FOR {$n\in\Vfix$}
                        \STATE \textsf{loss} $\gets \gL_{n_0}^{\precondgen}(\{l_{n'}, u_{n'}\}_{n'\in\Vfix})$
                        \STATE $c_n \gets c_n - \gamma_c \max\{|c_n|,\minstep\}\sgn(\nabla_{c_n} \mathsf{loss})$
                        \STATE $(l_n, u_n) \gets (c_n - \frac{s(u_n^{\valid} - l_n^{\valid})}{2}, c_n + \frac{s(u_n^{\valid} - l_n^{\valid})}{2})$
                        \STATE $(l_n, u_n) \gets (\max \{l_n, l_n^{\valid}\}, \min \{u_n, u_n^{\valid}\})$ 
                    \ENDFOR
                    \IF {$\gL_{n_0}^{\precondgen}(\{l_n, u_n\}_{n\in\Vfix}) < 0$}
                        \RETURN $\{l_n, u_n\}_{n\in\Vfix}$ \COMMENT{Find precondition fix}
                    \ENDIF
                    \STATE $s \gets \gamma_s \cdot s$
                \ENDFOR
                \RETURN ``failed'' \COMMENT{Failed to find precondition fix}
            \end{algorithmic}
        \end{algorithm}

        As shown in \Cref{alg:abs-opt}, we propose the technique of \textbf{abstraction optimization} to effectively and approximately solve this optimization.
        Our technique works iteratively.
        In the first iteration, we set span $s=1$, and in the subsequent iterations, we reduce the span $s$ exponentially as shown in Line 13 where hyperparameter $\gamma_s = 0.9$.
        Inside each iteration, for each node to impose precondition $n\in\Vfix$, we use the interval center $c_n = (l_n + u_n)/2$ as the optimizable variable and compute the \emph{sign} of its gradient: $\sgn(\nabla_{c_n} \mathsf{loss})$.
        We use this gradient sign to update each $c_n$ toward  reducing the loss value in Line 6.
        Then, we use $c_n$ and the span $s$ to recover the actual interval in Line 7 and clip $l_n$ and $u_n$ by the valid range $[l_n^{\valid}, u_n^{\valid}]$ in Line 8.
        At the end of this iteration, for updated $l_n$ and $u_n$, we compute $\gL_{n_0}^{\precondgen}(\{l_n, u_n\}_{n\in\Vfix})$ to check whether the precondition is a fix.
        If so, we terminate; otherwise, we proceed to the next iteration.
        We note that \emph{if the algorithm finds a precondition, the precondition is guaranteed to be a valid fix} by the soundness nature of our static analysis framework and the definition of $\gL_{n_0}^{\precondgen}$.
        When no feasible precondition is found within $\maxiter = 1000$ iterations, we terminate the algorithm and report ``failed to find the fix''.
        
        \begin{remark}
            The key ingredient in the technique is the gradient-sign-based update rule~(shown in Line 6), which is much more effective than normal gradient descent for two reasons. 
            (1)~Our update rule can get rid of gradient explosion and vanishing problems.
            For early optimization iterations, the span $s$ is large and interval bounds are generally coarse, resulting in too large or too small gradient magnitude.
            For example, the input range for \CodeIn{Log} could be $[1,10^{10}]$ where gradient can be $10^{-10}$, resulting in almost negligible gradient updates.
            In contrast, our update rule leverages the gradient sign, which always points to the correct gradient direction.
            The update step size in our rule is the maximum of current magnitude $|c_n|$ and $\minstep$ to avoid stagnation.
            (2)~Our update rule mitigates the gradient magnitude discrepancy of different $c_n$.
            At different locations, the nodes in DNNs can have diverse value magnitudes that are not aligned with their gradient magnitudes, making gradient optimization challenging.
            Therefore, we use this update rule to solve the challenge, where the update magnitude depends on the value magnitude ($|c_n|$) instead of gradient magnitude ($\nabla_{c_n} \mathsf{loss}$).
            We empirically compare our technique with standard gradient descent in \Cref{subsec:precond-gen-exp}.
        \end{remark}
        
    \begin{table*}[!t]
    \centering
    \caption{(RQ1) Results of task \circled{2}a~(failure-exhibiting \underline{\textbf{unit}} test generation) with \sysname and GRIST~\cite{yan2021exposing}.
    C is the number
    of runs where numerical failures are triggered in 10 repeated runs, T is 
    the average execution time per run, and $\Uparrow$T is the average time improvement achieved by \sysname compared to GRIST.}
    \vspace{-1em}
\resizebox{0.95\linewidth}{!}{
    \begin{tabular}{r|rrr|rr||r|rrr|rr||r|rrr|rr||r|rrr|rr}
    \toprule
     Case & \multicolumn{3}{c|}{\sysname} & \multicolumn{2}{c||}{GRIST} & Case & \multicolumn{3}{c|}{\sysname} & \multicolumn{2}{c||}{GRIST} & Case & \multicolumn{3}{c|}{\sysname} & \multicolumn{2}{c||}{GRIST} & Case & \multicolumn{3}{c|}{\sysname} & \multicolumn{2}{c}{GRIST} \\
     \cline{2-6}\cline{8-12}\cline{14-18}\cline{20-24}
     ID & C & T & $\Uparrow$T & C & T & ID & C & T & $\Uparrow$T & C & T & ID & C & T & $\Uparrow$T & C & T & ID & C & T & $\Uparrow$T & C & T \\
    \hline
1 & 10 & 9.01 & 1.20 X & 10 & 10.77 & 16b & 10 & 0.21 & 20.85 X & 10 & 4.42 & 32 & 10 & 0.06 & 27.93 X & 10 & 1.77 & 47 & 10 & 0.06 & 32.51 X & 10 & 1.87 \\
2a & 10 & 0.02 & 9.75 X & 10 & 0.24 & 16c & 10 & 0.25 & 17.54 X & 10 & 4.43 & 33 & 10 & 0.06 & 33.63 X & 10 & 1.91 & 48a & 10 & 0.38 & 3.06 X & 10 & 1.17 \\
2b & 10 & 0.03 & 614.68 X & 10 & 16.54 & 17 & 10 & 439.19 & $+\infty$ & 0 & - & 34 & 10 & 0.06 & 33.95 X & 10 & 1.90 & 48b & 10 & 0.15 & 7.12 X & 10 & 1.10 \\
3 & 10 & 0.02 & 432.11 X & 10 & 8.67 & 18 & 10 & 0.02 & 1040.46 X & 10 & 22.17 & 35a & 10 & 0.44 & 61.76 X & 10 & 27.33 & 49a & 10 & 0.49 & 41.09 X & 10 & 20.07 \\
4 & 10 & 0.01 & 1.00 X & 10 & 0.01 & 19 & 10 & 0.16 & 689.66 X & 10 & 107.78 & 35b & 10 & 0.45 & 819.02 X & 10 & 364.86 & 49b & 10 & 0.50 & 612.22 X & 10 & 307.24 \\
5 & 10 & 0.05 & 6.48 X & 10 & 0.34 & 20 & 10 & 0.16 & 3237.27 X & 10 & 511.06 & 36a & 10 & 0.44 & 41.80 X & 10 & 18.58 & 50 & 10 & 0.16 & 781.02 X & 10 & 126.80 \\
6 & 10 & 0.84 & 5.20 X & 10 & 4.38 & 21 & 10 & 0.16 & 259.73 X & 10 & 42.09 & 36b & 10 & 0.46 & 783.16 X & 10 & 362.41 & 51 & 10 & 1.88 & 671.55 X & 3 & 1263.04 \\
7 & 10 & 0.87 & 4.54 X & 10 & 3.96 & 22 & 10 & 0.94 & 1518.43 X & 10 & 1433.12 & 37 & 10 & 0.06 & 38.34 X & 10 & 2.39 & 52 & 10 & 0.15 & 336.37 X & 10 & 50.59 \\
8 & 10 & 0.86 & 4.63 X & 10 & 3.99 & 23 & 10 & 0.01 & 157.72 X & 10 & 1.88 & 38 & 10 & 0.06 & 34.50 X & 10 & 1.94 & 53 & 10 & 0.05 & 36.64 X & 10 & 1.92 \\
9a & 10 & 0.20 & 11.03 X & 10 & 2.22 & 24 & 10 & 0.81 & 40.72 X & 10 & 33.05 & 39a & 10 & 0.43 & 42.79 X & 10 & 18.30 & 54 & 10 & 0.05 & 36.76 X & 10 & 1.83 \\
9b & 10 & 0.14 & 14.46 X & 10 & 2.09 & 25 & 10 & 0.04 & 1271.88 X & 10 & 44.70 & 39b & 10 & 0.43 & 843.66 X & 10 & 362.22 & 55 & 10 & 0.82 & 44.63 X & 10 & 36.63 \\
10 & 10 & 0.17 & 228.42 X & 10 & 39.64 & 26 & 10 & 0.05 & 37.96 X & 10 & 2.00 & 40 & 10 & 0.04 & 1995.27 X & 10 & 85.97 & 56 & 10 & 0.06 & 35.04 X & 10 & 1.93 \\
11a & 10 & 0.15 & 27.58 X & 10 & 4.26 & 27 & 10 & 0.01 & 185.61 X & 10 & 1.91 & 41 & 10 & 0.04 & 1967.23 X & 10 & 86.36 & 57 & 10 & 0.01 & 177.45 X & 10 & 1.88 \\
11b & 10 & 0.13 & 34.75 X & 10 & 4.38 & 28a & 10 & 24.37 & -13.30 X & 10 & 1.83 & 42 & 10 & 0.05 & 1934.84 X & 10 & 87.89 & 58 & 10 & 0.83 & 12.01 X & 10 & 9.95 \\
11c & 10 & 0.11 & 4499.86 X & 10 & 516.13 & 28b & 10 & 24.17 & 7.28 X & 10 & 176.02 & 43a & 10 & 0.48 & 35.63 X & 10 & 16.96 & 59 & 10 & 0.02 & 105.40 X & 10 & 1.94 \\
12 & 10 & 0.26 & 135.94 X & 10 & 34.69 & 28c & 10 & 0.12 & 8.69 X & 10 & 1.02 & 43b & 10 & 0.45 & 4008.93 X & 10 & 1800.00 & 60 & 10 & 0.15 & 221.97 X & 10 & 34.19 \\
13 & 10 & 0.01 & 1.10 X & 10 & 0.01 & 28d & 10 & 0.12 & 1518.28 X & 10 & 176.02 & 44 & 10 & 0.27 & 579.29 X & 10 & 155.38 & 61 & 10 & 0.35 & 53.29 X & 10 & 18.78 \\
14 & 10 & 0.80 & 107.96 X & 10 & 86.23 & 29 & 10 & 0.89 & 16.83 X & 10 & 14.98 & 45a & 10 & 0.16 & 417.25 X & 10 & 68.08 & 62 & 10 & 1.85 & 72.19 X & 10 & 133.62 \\
15 & 10 & 1.71 & 5.95 X & 10 & 10.18 & 30 & 10 & 0.16 & 222.12 X & 10 & 35.61 & 45b & 10 & 0.88 & 14.69 X & 10 & 12.98 & 63 & 10 & 2.06 & 117.12 X & 10 & 240.68 \\
\cline{19-24}
16a & 10 & 0.12 & 34.24 X & 10 & 4.02 & 31 & 10 & 3.13 & 2.41 X & 3 & 7.54 & 46 & 10 & 0.01 & 168.39 X & 10 & 1.88 & \textbf{Tot: 79} & \textbf{790} & \textbf{6.66} & \textbf{17.32 X} & \textbf{766} & \textbf{115.30} \\
    \bottomrule
    \end{tabular}
}
    \label{tab:unittest}
    \vspace{-1.0em}
\end{table*}

\section{Experimental Evaluation}

    \label{sec:exp}

    We conduct a systematic experimental evaluation to answer the following research questions.

    \begin{enumerate}[leftmargin=*,label={\textbf{RQ\arabic*}}]
        \item
        For tasks already supported by existing state-of-the-art~(SOTA)  tools~(tasks \circled{1} and \circled{2}a), how much more effective and efficient is \sysname compared to these SOTA  tools?
    
        \item 
        For feasibility confirmation via \emph{generating failure-exhibiting system tests}~(task \circled{2}), how much more effectively and efficiently can \sysname confirm potential numerical defects compared to baseline approaches?
        
        \item 
        For \emph{suggesting fixes}~(task \circled{3}), how much more efficient and effective is \sysname in terms of guarding against numerical failures compared to baseline approaches and developers' fixes, respectively?
        
    \end{enumerate}
    
    For RQ1, we compare \sysname with all SOTA tools.
    For RQ2 and RQ3, \sysname is the first approach to the best of our knowledge, so we compare \sysname with baseline approaches (constructed by leaving our novel techniques out of  \sysname) and developers' fixes.
    We conduct the evaluation on the GRIST benchmarks~\cite{yan2021exposing}, being the largest dataset of real-world DNN numerical defects to our knowledge. The benchmarks contain 63 real-world DL programs with numerical defects collected from previous studies and GitHub.
    Each program contains a DNN architecture, and each architecture has one or more numerical defects.
    There are 79 real numerical defects in total.

    We perform our evaluation on a Linux workstation with a 24-core Xeon E5-2650 CPU running at 2.20 GHz.
    Throughout the evaluation, we stop the execution after reaching \SI{30}{min} limit by following the evaluation setup by the most recent related work~\cite{yan2021exposing}.

    \subsection{RQ1: Comparison with SOTA Tools}
        \label{subsec:static-detector-exp}
        \label{subsec:compare-with-sota}
        
        
        For two tasks, existing tools can provide automatic support: potential-defect detection~(task \circled{1}) where the SOTA tool is DEBAR~\cite{zhang2020detecting}, and failure-exhibiting unit test generation~(task \circled{2}a) where the SOTA tool is GRIST~\cite{yan2021exposing}.
        We compare \sysname with these tools on their supported tasks, respectively.
        
        \paragraph{Comparison with DEBAR}
        \sysname successfully detects all 79 true defects and DEBAR detects only 48 true defects according to both our evaluation and the literature~\cite{yan2021exposing}.
        Hence, \sysname detects 64.58\% more true defects than DEBAR.
        In terms of efficiency, DEBAR and \sysname have similar running time, and both finish in \SI{3}{s} per case.
        
        
        We manually inspect the cases where DEBAR fails but \sysname succeeds.
        They correspond to DL programs written with the PyTorch library, which generates dynamic computational graphs that DEBAR cannot handle.
        In contrast, \sysname provides effective static analysis support for dynamic computational graphs thanks to our backward fine-grained node labeling technique~(\Cref{subsec:static-analysis}) that is capable of disambiguating the control flow within dynamic graphs.

        
        \begin{table*}[!t]
    \centering
    \vspace{0.2em}
    \caption{(RQ2) Results of task \circled{2}~(failure-exhibiting \underline{\textbf{system}} test generation) with \sysname and Random~(baseline).
    C is the total number of runs where numerical failures are triggered in 10 repeated runs.
    T is the average execution time per run.}
    \vspace{-1em}
\resizebox{0.95\linewidth}{!}{
    \begin{tabular}{r|rr|rr||r|rr|rr||r|rr|rr||r|rr|rr||r|rr|rr}
    \toprule
    Case & \multicolumn{2}{c|}{\sysname} & \multicolumn{2}{c||}{Random} & 
    Case & \multicolumn{2}{c|}{\sysname} & \multicolumn{2}{c||}{Random} & 
    Case & \multicolumn{2}{c|}{\sysname} & \multicolumn{2}{c||}{Random} &
    Case & \multicolumn{2}{c|}{\sysname} & \multicolumn{2}{c||}{Random} &
    Case & \multicolumn{2}{c|}{\sysname} & \multicolumn{2}{c}{Random} \\
    \cline{2-5}
    \cline{7-10}
    \cline{12-15}
    \cline{17-20}
    \cline{22-25}
    ID & C & T & C & T &
    ID & C & T & C & T &
    ID & C & T & C & T &
    ID & C & T & C & T &
    ID & C & T & C & T \\
    \hline
    
1 & 0 & 9.01 & 0 & 1806.13 & 13 & 10 & 0.01 & 10 & 0.01 & 27 & 10 & 0.01 & 10 & 0.01 & 38 & 1 & 0.13 & 0 & 1800.65 & 49b & 10 & 0.50 & 8 & 364.09 \\
2a & 10 & 0.03 & 10 & 0.06 & 14 & 10 & 12.60 & 10 & 0.50 & 28a & 0 & 24.37 & 0 & 1920.29 & 39a & 10 & 0.43 & 1 & 1623.72 & 50 & 10 & 4.89 & 10 & 0.16 \\
2b & 10 & 0.03 & 10 & 0.06 & 15 & 0 & 1.71 & 0 & 2107.24 & 28b & 0 & 24.17 & 0 & 1911.26 & 39b & 10 & 0.43 & 8 & 364.10 & 51 & 10 & 49.12 & 10 & 2.10 \\
3 & 10 & 0.02 & 10 & 0.05 & 16a & 10 & 0.12 & 10 & 0.75 & 28c & 10 & 0.12 & 10 & 0.53 & 40 & 10 & 0.06 & 10 & 0.02 & 52 & 10 & 4.87 & 10 & 0.15 \\
4 & 10 & 0.01 & 10 & 0.01 & 16b & 10 & 0.21 & 0 & 1834.44 & 28d & 10 & 0.12 & 10 & 0.48 & 41 & 10 & 0.06 & 10 & 0.02 & 53 & 10 & 0.07 & 10 & 0.03 \\
5 & 10 & 0.05 & 10 & 0.06 & 16c & 10 & 0.25 & 0 & 1831.67 & 29 & 10 & 0.89 & 10 & 11.96 & 42 & 10 & 0.06 & 10 & 0.02 & 54 & 10 & 0.07 & 10 & 0.02 \\
6 & 10 & 0.84 & 10 & 12.42 & 17 & 10 & 549.98 & 10 & 235.11 & 30 & 10 & 4.88 & 10 & 0.14 & 43a & 10 & 0.48 & 1 & 1623.71 & 55 & 10 & 0.82 & 10 & 12.79 \\
7 & 10 & 0.87 & 10 & 12.51 & 18 & 10 & 0.02 & 10 & 0.05 & 31 & 10 & 14.62 & 10 & 9.31 & 43b & 10 & 0.45 & 9 & 184.14 & 56 & 10 & 0.07 & 10 & 0.03 \\
8 & 10 & 0.86 & 10 & 12.37 & 19 & 10 & 4.88 & 10 & 0.16 & 32 & 10 & 0.08 & 10 & 0.03 & 44 & 10 & 0.27 & 10 & 1.36 & 57 & 10 & 0.01 & 8 & 360.01 \\
9a & 10 & 0.20 & 7 & 541.25 & 20 & 10 & 4.88 & 10 & 0.14 & 33 & 10 & 0.07 & 10 & 0.02 & 45a & 10 & 4.89 & 10 & 0.15 & 58 & 10 & 0.83 & 10 & 12.28 \\
9b & 10 & 0.14 & 10 & 1.39 & 21 & 10 & 4.89 & 10 & 0.14 & 34 & 10 & 0.42 & 10 & 0.20 & 45b & 10 & 0.88 & 10 & 12.27 & 59 & 10 & 0.02 & 10 & 0.05 \\
10 & 10 & 4.90 & 10 & 0.16 & 22 & 10 & 500.10 & 0 & 1801.60 & 35a & 10 & 0.44 & 10 & 4.01 & 46 & 10 & 0.01 & 10 & 0.01 & 60 & 10 & 4.88 & 10 & 0.16 \\
11a & 10 & 0.15 & 10 & 0.72 & 23 & 10 & 0.01 & 10 & 0.01 & 35b & 10 & 0.45 & 10 & 4.22 & 47 & 10 & 0.08 & 10 & 0.03 & 61 & 10 & 9.84 & 10 & 0.93 \\
11b & 10 & 0.13 & 10 & 0.76 & 24 & 10 & 0.81 & 10 & 12.52 & 36a & 10 & 0.44 & 1 & 1623.76 & 48a & 10 & 9.89 & 10 & 0.90 & 62 & 10 & 48.86 & 10 & 2.73 \\
11c & 10 & 0.11 & 10 & 0.74 & 25 & 10 & 0.04 & 10 & 0.15 & 36b & 10 & 0.46 & 3 & 1263.79 & 48b & 10 & 4.88 & 10 & 0.15 & 63 & 10 & 49.06 & 10 & 2.15 \\
\cline{21-25}
12 & 10 & 0.26 & 10 & 0.72 & 26 & 10 & 0.07 & 10 & 0.03 & 37 & 2 & 0.07 & 2 & 1440.50 & 49a & 10 & 0.49 & 1 & 1623.88 & \textbf{Tot: 79} & \textbf{733} & \textbf{17.31 (19.30X)} & \textbf{649} & \textbf{334.14} \\

\bottomrule
    \end{tabular}
}
    \label{tab:systest}
    \vspace{-1.9em}
\end{table*}
        
        \paragraph{Comparison with GRIST}
        Results are shown in \Cref{tab:unittest}.
        Since both \sysname and GRIST have a randomness component where \sysname uses random initialization and GRIST relies on DNN's randomly initialized weights, we repeat both approaches for 10 runs, record the total number of times where a failure-exhibiting unit test is generated, and the average execution time per run.
        \sysname succeeds in \emph{all} cases and \emph{all} repeated runs, and GRIST fails to generate such unit test in 24 out of 790 runs~(i.e., 96.96\% success rate).
        \sysname has \SI{6.66}{s} average execution time and is 17.32X faster than~GRIST.
        
        The superior effectiveness and efficiency of \sysname are largely due to the existence of random initialization as introduced in \Cref{subsec:system-test-gen}.
        We observe that since GRIST always takes initial model weights and inference input as the starting point to update from, the generated unit test is just slightly changed from initial weights and input, being hard to expose the target numerical defect.
        In contrast, \sysname uses random initialization to explore a much larger space and combines gradient-based optimization to locate the failure-exhibiting instances from the large space.
        We also evaluate  the pure random strategy that  uses only random initialization without gradient-based optimization, and such strategy fails in 30 runs, being inferior to both \sysname and GRIST, implying that both random  initialization and gradient-based optimization are important.
        Among all the 79 cases, \sysname is slower than GRIST on only one case~(28a).
        For this case, we find the default inference input loaded by the DNN program~(used by GRIST) is not far from a failure-exhibiting one, but a randomly sampled inference input~(used by \sysname) is usually far from that.
        Hence, \sysname takes more iterations to find a failure-exhibiting inference input by the nature of gradient-based optimization.

    \subsection{RQ2: Feasibility Confirmation via System Test Generation}
    
        \label{subsec:system-test-gen-exp}
        

        
        In task \circled{2}, \sysname confirms the feasibility of potential numerical defects by generating failure-exhibiting system tests.
        
        \paragraph{Baseline}
        Since \sysname is the first approach for this task,
        we do not compare with existing literature and propose one random-based approach~(named ``Random'' hereinafter)  as the baseline.
        In ``Random'', we first generate a  failure-exhibiting unit test with random sampling.
        If there is any sample that triggers a  failure, we stop and keep the inference example part as the desired $\vxinfer$.
        Then, we generate $\vxtrain$ again by random sampling.
        If any sample, when used for training, could induce model weights $\vw$ that cause a numerical failure when using $\vxinfer$ as the inference input, we keep that sample as $\vxtrain$ and terminate.
        If and only if both $\vxinfer$ and $\vxtrain$ are found, we count this run as a ``success'' one for ``Random''.
        
        For each defect, due to the randomness of the model's initial weights, we repeat both \sysname and ``Random'' for $10$ runs.
        Both approaches use the same set of random seeds.
        
        \paragraph{Evaluation Result}
        Results are in \Cref{tab:systest}.
        We observe that \sysname succeeds in $733 / (79\times 10) = 92.78\%$ runs and the baseline ``Random'' succeeds in $649 / (79\times 10) = 82.15\%$ runs.
        Moreover, \sysname spends only \SI{17.31}{s} time on average per run, which is a 19.30X speedup compared to ``Random''.
        We also observe that \sysname is more reliable across repeated runs. There are only 6 cases with unsuccessful repeated runs in \sysname, but there are 19 such cases in ``Random''. 
        Hence, \sysname is substantially more effective, efficient, and reliable for generating system tests than the baseline.
        
        \paragraph{Discussion}
        The high effectiveness of \sysname mainly comes from the advantage of gradient-guided search compared with random search.
        As described in \Cref{subsec:system-test-gen}, \sysname leverages both first-order gradients~(in step a) and second-order derivatives~(in step b) to guide the search of system tests.
        In contrast, ``Random'' uses random sampling hoping that failure-exhibiting training examples can emerge after sufficient sampling.
        Hence, when such training examples are relatively rare in the whole valid input space, ``Random'' is less effective.
        We conduct an ablation study~(in \Cref{adxsec:system-testgen-by-stage}) for showing that \sysname improves over ``Random'' in both steps inside \sysname.
        
        \paragraph{Failing-Case Analysis}
        We study all six defects where \sysname may fail and have the following findings.
        (1)~For four defects~(Case IDs 1, 15, 37, and 38), the architecture is challenging for gradient-based optimization, e.g., due to the \CodeIn{Min}/\CodeIn{Max}/\CodeIn{Softmax} operators that provide little or no gradient information.
        We leave it as future work to solve these cases, likely in need of  dynamically detecting operators with vanishing gradients and reconstructing the gradient flow.
        (2)~Two defects~(Case IDs 28a and 28b) correspond to those caused by \CodeIn{Div} operators where only a close-to-zero divisor can trigger a  numerical failure. 
        Hence, for operators with narrow invalid ranges, RANUM may fail to generate failure-exhibiting system tests. 
        

    \subsection{RQ3: Fix Suggestion}
    
        \label{subsec:precond-gen-exp}
        
        In task \circled{3}, \sysname suggests fixes for numerical defects.
        We compare \sysname with fixes generated by baseline approaches and developers' fixes.
        
        \subsubsection{\textbf{Comparison between \sysname and Baselines}}
            \begin{table}[!t]
    \centering
    \caption{(RQ3) Results of task \circled{3}~(fix suggestion) under three  imposing location specifications
    with \sysname and two baselines~(\sysname-E and GD). \# is the  number of fixes found. ``Time (s)'' is the total running time for all 79 cases.
    }
    \vspace{-1.0em}
\resizebox{0.9\linewidth}{!}{
    \begin{tabular}{c|cc|cc|cc}
        \toprule
        Imposing & \multicolumn{2}{c|}{\sysname} & \multicolumn{2}{c|}{\sysname-E} & \multicolumn{2}{c}{GD} \\
        \cline{2-7}
        Locations & \# & Time (s) & \# & Time (s) & \# & Time (s) \\
        \hline
        Weight + Input & \textbf{79} & \textbf{54.23} & 78 & 540.13 & 57 & 188.63 \\
        Weight & \textbf{72} & \textbf{58.47} & 71 & 581.86 & 43 & 219.28 \\
        Input & \textbf{37} & \textbf{924.74} & \textbf{37} & 3977.30 & 29 & 952.19 \\
        \bottomrule
    \end{tabular}

}
    \label{tab:precond-summary}
\end{table}


            
            \sysname is the first approach for this task,
            and we propose two baseline approaches to compare with. 
            (1)~\sysname-E: this approach changes the abstraction domain of \sysname from interval with tensor partitioning to standard interval.
            To some degree, \sysname-E represents the methodology of conventional static analysis tools that use standard interval domain for abstraction and search of effective fixes. 
            (2)~GD: this approach uses standard gradient descent for optimization instead of the abstraction optimization technique in \sysname.
            
            \paragraph{Evaluation Protocol}
            We evaluate whether each approach can generate fixes that eliminate \emph{all} numerical defects for the DNN architecture under analysis  given imposing locations.
            We consider three types of locations: on both weight and input nodes, on only weight nodes, and on only input nodes.
            In practice, model providers can impose fixes on weight nodes by clipping weights after a model is trained; and users can impose fixes on input nodes by clipping their inputs before loading them into the model.
            Since all approaches are deterministic, for each case we run only once.
            We say that the fix eliminates all numerical defects if and only if (1)~the \sysname static analysis framework cannot detect any defects from the fixed architecture; and (2)~1,000 random samples cannot trigger any numerical failures after imposing the fix.
            
            \paragraph{Evaluation Result}
            We report the statistics, including the number of successful cases among all the 79 cases and the total running time, in \Cref{tab:precond-summary}.
            From the table, we observe that on all the three imposing location settings, \sysname always succeeds in most cases and spends much less time.
            For example, when fixes can be imposed on both weights and input nodes, \sysname succeeds on \emph{all} cases with a total running time \SI{54.23}{s}.
            In contrast, \sysname-E requires $>10 \times$ time, and GD succeeds in only $72.15\%$ cases.
            Hence, \sysname is substantially more effective and efficient for suggesting fixes compared to baseline approaches.

            Since \sysname is based on iterative refinement~(see \Cref{alg:abs-opt}), we study the number of iterations needed for finding the fix.
            When fixes can be imposed on both weight  and input nodes, where \sysname succeeds on all the 79 cases, the average number of iterations is $29.80$, the standard deviation is $14.33$, the maximum is $53$, and the minimum is $2$.
            Hence, when \sysname can find the fix, the number of iterations is small, coinciding with the small total running time \SI{54.23}{s}.
            
            \paragraph{Discussion}
            The two baseline approaches can be viewed as ablated versions of \sysname.
            Comparing \sysname and GD, we conclude that the technique of abstraction optimization substantially improves the effectiveness and also improves the efficiency.
            Comparing \sysname and \sysname-E, we conclude that the interval abstraction with tensor partitioning as the abstraction domain substantially improves the efficiency and also improves the effectiveness.
            
            From \Cref{tab:precond-summary}, it is much easier to find the fix when imposing locations are weight nodes compared to input nodes.
            Since model providers can impose fixes on weights and users impose on inputs, this finding implies that fixing numerical defects on the providers' side may be more effective than on the users' side.

        \subsubsection{\textbf{Comparison between \sysname and Developers' Fixes}}
                
            We conduct an empirical study to compare the fixes generated by \sysname and by the developers.
            
            \paragraph{Evaluation Protocol}
            We manually locate GitHub repositories from which the GRIST benchmarks are constructed.
            Among the 79 cases, we find the repositories for 53 cases on GitHub and we study these cases.
            We locate the developers' fixes of the numerical defects by looking at issues and follow-up pull requests.
            Since \sysname suggests different fixes for different imposing locations, for each case we first determine the imposing locations from the developer's fix, and then compare with \sysname's fix for these locations.

            \sysname fixes are on the computational graph and developers’ fixes are in the source code, so we determine to conduct code-centered comparison: 
            \sysname fixes are considered feasible only when the fixes can be easily implemented by code~(within 10 lines of code) given that developers’ fixes are typically small, usually in 3-5 lines of code. 
            In particular, our comparison is based on two criteria: 
            (1)~which fix is sound on any valid input; 
            (2)~if both are sound, which fix hurts less to model performance and utility (based on the span of imposed precondition, the larger span the less hurt). 
            Two authors independently classify the comparison results for each case and discuss  the results to reach a consensus.
            
            \paragraph{Results}
            We categorize the comparison results 
            as below.
            \begin{enumerate}[leftmargin=*,label=\Alph*]
                \item \emph{(30 cases) Better than developers' fixes or no available developer's fix.}
                Developers either propose no fixes or use heuristic fixes, such as reducing the learning rate or using the mean value to reduce the variance.
                These fixes may work in practice but are unsound, i.e., cannot rigorously guarantee the elimination of the numerical defect for any training or inference data.
                In contrast, \sysname generates better fixes since these fixes rigorously eliminate the defect.
                
                \item \emph{(7 cases) Equivalent to developers' fixes.}
                Developers and \sysname suggest equivalent or highly similar fixes.
                
                \item \emph{(13 cases) No need to fix.}
                For these cases, there is no need to fix the numerical defect in the given architecture.
                There are mainly three reasons. 
                (1)~The DNN is used in the whole project with fixed weights or test inputs.
                As a result, although the architecture contains defects, no  system failure can be caused.
                (2)~The architecture is injected a defect as a test case for automatic tools, such as a test architecture in the \CodeIn{TensorFuzz}~\cite{odena2019tensorfuzz} repository.
                (3)~The defect can be hardly exposed in practice. 
                For example, the defect is in a \CodeIn{Div} operator where the divisor needs to be very close to zero to trigger a divide-by-zero failure, but such situation hardly happens in practice since the divisor is randomly initialized.
                
                \item \emph{(3 cases) Inferior than developers' fixes or \sysname-generated fixes are impractical.}
                In two cases, \sysname-generated fixes are inferior to developers' fixes.
                Human developers observe that the defective operator is \CodeIn{Log}, and its input is non-negative.
                So they propose to add $10^{-6}$ to the input of \CodeIn{Log} as the fix.
                In contrast, \sysname can  generate only a clipping-based fix, e.g., clipping the input if it is less than $10^{-6}$. 
                When the input is small, \sysname's fix interrupts the gradient flow from output to input while the human's fix maintains it.
                As a result, the human's fix does less hurt to the model's trainability and is better than \sysname's fix.
                In another  case, the \sysname-generated fix imposes a small span for some model weights~(less than $0.1$ for each component of that weight node).
                Such a small weight span strongly limits the model's expressiveness and utility.
                We leave it as the future work to solve these limitations.
            \end{enumerate}
            
            From the comparison results, we can conclude that for the 40 cases where numerical defects are needed to be fixed~(excluding case C),  \sysname suggests equivalent or better fixes than human developers in 37 cases.
            Therefore, \sysname is comparably effective as human developers in terms of suggesting numerical-defect fixes, and is much more efficient since \sysname is an automatic approach.
            

            \paragraph{Guidelines for Users}
            We discuss two practical questions for \sysname users.
            (1)~\emph{Does \sysname hurt model utility, e.g., inference accuracy?}
            If no training or test data ever exposes a numerical defect, \sysname does not confirm a defect and hence no fix is  generated and there is no hurt to the utility. 
            If \sysname confirms numerical defects, whether the fix affects the utility depends on the precondition-imposing locations.
            If imposing locations can be freely selected, \sysname tends to impose the fix right before the vulnerable operator, and hence the fix does not reduce inference performance. 
            The reason is that the fix  changes~(by clipping) the input only when the input falls in the invalid range of the vulnerable operator. 
            In practice, if the imposing locations cannot be freely selected and need to follow developers’ requirements, our preceding empirical study shows that, in only 3 out of 40 cases, compared with developers' fixes, our fixes incur larger hurt to the inference or training performance of the architecture.
            (2)~\emph{Should we always apply \sysname to fix any architecture?}
            We can always apply \sysname to fix any architecture since \sysname fixes do not visibly alter the utility in most cases. 
            Nonetheless, in deployment, we recommend first using \sysname to confirm defect feasibility. 
            If there is no such failure-exhibiting system test, we may not need to fix the architecture; otherwise, we use \sysname to generate fixes.

    
\section{Related Work}


    \paragraph{Understanding and Detecting Defects in DNNs}
    Discovering and mitigating defects  and failures in DNN based systems is an important research topic~\cite{zhang2018empirical,pham2020problems,kloberdanz2022deepstability}.
    Following the taxonomy in previous work~\cite{humbatova2020taxonomy}, DNN defects are at four levels from bottom to top. 
    (1)~Platform-level defects. Defects can exist in real-world DL compilers and libraries. 
    Approaches exist for understanding, detecting, and testing against these defects~\cite{wang2020deep, DBLP:conf/sigsoft/ShenM0TCC21,tambon2021silent,xie2021leveraging}. 
    (2)~Architecture-level defects.  \emph{Our work focuses on numerical defects, being one type of architecture-level defects.} 
    Automatic detection and localization approaches~\cite{wardat2021deeplocalize,liu2021detecting} exist for other architecture-level defects such as suboptimal structure, activation function, and initialization and shape mismatch~\cite{hattori2020semi}.
    (3)~Model-level defects. Once a model is trained, its defects can be viewed as violations of desired properties as discussed by Zhang et al.~\cite{zhang2020machine}.
    Some example defects are correctness~\cite{tizpaz2020detecting,guerriero2021operation}, robustness~\cite{wang2021robot}, and fairness~\cite{zhang2020white} defects.
    (4)~Interface-level defects. DNN-based systems, when deployed as services, expose interaction interfaces to users where defects may exist, as shown by empirical studies on real-world systems~\cite{humbatova2020taxonomy,wan2021machine,wan2022automated}.
    
    \paragraph{Testing and Debugging for DNNs}
    A rich body of work exists for testing and debugging DNN defects~\cite{zhang2020machine}.
    Some representatives are DeepXplore~\cite{pei2017deepxplore} and DeepGauge~\cite{ma2018deepgauge}.
    Recent work enables automatic model debugging and repair via consistency checking~\cite{xiao2021self}, log checking~\cite{zhang2021autotrainer}, spectrum analysis~\cite{qi2021archrepair}, or analyzer-guided synthesis~\cite{sotoudeh2021provable}.
    
    \paragraph{DNN Static Analysis}
    Another solution for eliminating DNN defects is conducting static analysis to rigorously guarantee the non-existence of defects~\cite{li2020sok,albarghouthi2021introduction}.
    Although DNNs essentially conduct numerical computations, traditional tools of numerical analysis~\cite{gurfinkel2015seahorn,singh2017fast} are inefficient for DNN analysis due to lack of support for multi-dimensional tensor computations.
    Recently, static analysis tools customized for DNNs are emerging, mainly focusing on proposing tighter abstractions~\cite{gehr2018ai2,muller2022prima,paulsen2022linsyn} or incorporating abstractions into training to improve robustness~\cite{li2019robustra,mirman2018differentiable,zhang2021robustness}.
    Besides robustness, static analysis has also been applied to rigorously bound model difference~\cite{paulsen2020neurodiff}.
    Our approach includes a static analysis framework customized for numerical-defect detection and fixing.
    
    \paragraph{Detecting and Exposing Numerical Defects in DNNs}
    Despite the widespread existence of numerical defects in real-world DNN-based systems~\cite{zhang2018empirical,humbatova2020taxonomy,kloberdanz2022deepstability}, only a few automatic approaches exist for detecting and exposing these defects.
    To the best of our knowledge, DEBAR~\cite{zhang2020detecting} and GRIST~\cite{yan2021exposing} are the only two approaches.
    We discuss and compare \sysname with both approaches extensively in \Cref{sec:method,sec:exp}. 
\section{Conclusion}
    In this paper, we have presented a novel automatic approach  named \sysname for reliability assurance of DNNs against numerical defects.
    \sysname supports detection of potential numerical defects, confirmation
of potential-defect feasibility, and suggestion of defect fixes. 
     \sysname includes multiple novel extensions and optimizations upon existing tools, and includes three novel techniques.
    Our extensive evaluation on real-world DNN architectures has demonstrated high effectiveness and efficiency of \sysname
    compared to both the state-of-the-art approaches and developers' fixes.

    \paragraph{Data Availability}
    All artifacts including the tool source code and experiment logs are available and actively maintained at \CodeIn{\url{https://github.com/llylly/RANUM}}.
    \ifnum\arxiv=0
    The supplemental materials containing proofs, hyperparameters, and experiments are available in the arXiv version~\cite{li2023reliability}. 
    \fi

\section*{Acknowledgements}

    This work is sponsored by the National Natural Science Foundation of China under Grant No. 62161146003, the National Key Research and Development Program of China under Grant No. 2019YFE0198100, the Innovation and Technology Commission of HKSAR under Grant No. MHP/055/19, and the Tencent Foundation/XPLORER PRIZE.

\newpage
\balance
{
\small
\bibliographystyle{IEEEtranSN}
\bibliography{bib}
}

\clearpage

\appendix




\subsection{List of Supported Operators}

    \label{adxsec:supported-op}
    
    In this section, we provide a list of the 84 supported operators in the \sysname static analysis framework.
    These operators cover common operators used in DNNs. 
    To the best of our knowledge, \sysname provides abstractions for the largest number of operator types among existing DNN abstraction frameworks.
    We believe that the framework can be further extended to support other types of analysis, testing, and debugging for general DNN architectures or models such as robustness analysis and fairness testing.
    
    In particular, when compared to the state-of-the-art tool DEBAR~\cite{zhang2020detecting}, the underlined 17 operators are newly supported by \sysname.
    
    
    \noindent
    {\scriptsize \hspace{-0.3cm}
    \begin{tabular}{p{2.1cm}p{3.4cm}p{2.2cm}}
        \SmallCodeIn{Sub} & \SmallCodeIn{Add} & \SmallCodeIn{Mul} \\
        \SmallCodeIn{Div} & \SmallCodeIn{Pow} & \SmallCodeIn{MatMul} \\ 
        \underline{\SmallCodeIn{Gemm}} & \underline{\SmallCodeIn{MaxPool}} & \underline{\SmallCodeIn{GlobalMaxPool}} \\
        \SmallCodeIn{GlobalAveragePool} & \SmallCodeIn{Cos} & \SmallCodeIn{Sin} \\
        \SmallCodeIn{AveragePool} & \SmallCodeIn{Conv} & \SmallCodeIn{ConvTranspose} \\
        \underline{\SmallCodeIn{Pad}} & \underline{\SmallCodeIn{Reciprocal}} & \SmallCodeIn{Sqrt} \\
        \SmallCodeIn{Tanh} & \SmallCodeIn{Relu} & \SmallCodeIn{Softplus} \\
        \SmallCodeIn{LeakyRelu} & \underline{\SmallCodeIn{Softsign}} & \SmallCodeIn{Sigmoid} \\
        \SmallCodeIn{Neg} & \SmallCodeIn{Exp} & \SmallCodeIn{Log} \\
        \SmallCodeIn{Softmax} & \SmallCodeIn{LogSoftmax} & \SmallCodeIn{Abs} \\
        \SmallCodeIn{Ceil} & \SmallCodeIn{Floor} & \SmallCodeIn{Sign} \\
        \SmallCodeIn{Reshape} & \SmallCodeIn{Flatten} & \SmallCodeIn{Transpose} \\
        \SmallCodeIn{Shape} & \SmallCodeIn{Cast} & \SmallCodeIn{Slice} \\
        \underline{\SmallCodeIn{Gather}} & \underline{\SmallCodeIn{GatherND}} & \SmallCodeIn{Squeeze} \\
        \SmallCodeIn{Unsqueeze} & \underline{\SmallCodeIn{ScatterElements}} & \SmallCodeIn{Expand} \\
        \SmallCodeIn{Concat} & \SmallCodeIn{Split} & \SmallCodeIn{Identity} \\
        \SmallCodeIn{ConstantOfShape} & \underline{\SmallCodeIn{RandomNormalLike}} & \underline{\SmallCodeIn{RandomUniformLike}} \\
        \underline{\SmallCodeIn{RandomNormal}} & \underline{\SmallCodeIn{RandomUniform}} & \SmallCodeIn{Range} \\
        \SmallCodeIn{Constant} & \SmallCodeIn{Less} & \SmallCodeIn{LessOrEqual} \\
        \SmallCodeIn{Greater} & \SmallCodeIn{GreaterOrEqual} & \SmallCodeIn{Equal} \\
        \SmallCodeIn{Not} & \SmallCodeIn{Or} & \SmallCodeIn{Min} \\
        \SmallCodeIn{Max} & \SmallCodeIn{Clip} & \SmallCodeIn{Sum} \\
        \SmallCodeIn{ReduceMin} & \SmallCodeIn{ReduceMax} & \SmallCodeIn{ReduceSum} \\
        \SmallCodeIn{ReduceMean} & \SmallCodeIn{ArgMax} & \SmallCodeIn{ArgMin} \\
        \SmallCodeIn{Tile} & \underline{\SmallCodeIn{NegativeLogLikelihoodLoss}} & \underline{\SmallCodeIn{Loop}} \\
        \SmallCodeIn{SequenceInsert} & \SmallCodeIn{BatchNormalization} & \SmallCodeIn{OneHot} \\
        \underline{\SmallCodeIn{NonZero}} & \underline{\SmallCodeIn{Resize}} & \SmallCodeIn{ReduceProd} \\
        \SmallCodeIn{ReduceSumSquare} & \SmallCodeIn{IsInf} & \SmallCodeIn{IsNaN} \\
    \end{tabular}
    }

\subsection{List of Operators with Potential Numerical Defects}

    \label{adxsec:list-numerical-defects}
    
    In this section, we provide a full list of DNN operators that may contain numerical defects, along with their invalid ranges $\gI_{n_0,\mathsf{invalid}}$~(see definition of the numerical defect in \Cref{def:numerical-defect}), respectively.
    In the table, $U_{\min}$ and $U_{\max}$ stand for the minimum and maximum positive number of the input tensor's data type, respectively.
    
    {
    \centering \vspace{0.5em}
    \resizebox{\linewidth}{!}{
    \begin{tabular}{cc}
        \toprule
        Op. Type & $\gI_{n_0,\mathsf{invalid}}$ \\
        \hline
        \CodeIn{Pow} & $[-U_{\min}, U_{\min}] \times (-\infty, -U_{\min}]$ \\
        \CodeIn{Div} & $\sR \times [-U_{\min}, U_{\min}]$ \\
        \CodeIn{Reciprocal} & $[-U_{\min}, U_{\min}]$ \\
        \CodeIn{Sqrt} & $(-\infty, U_{\min}]$ \\
        \CodeIn{Exp} & $[\ln U_{\max}, \infty)$ \\
        \CodeIn{Log} & $(-\infty, U_{\min}]$ \\
        \CodeIn{Range} & $\sR \times \sR \times [-U_{\min}, U_{\min}]$ \\
        \CodeIn{NegativeLogLikelihoodLoss} & $[0,0]$ for number of non-zero cells using mean reduction \\
        \bottomrule
    \end{tabular}
    }
    }
    
\subsection{Detail Description of \sysname Static Analysis Framework}
    
    \label{adxsec:static-analysis}
    
    We present the omitted details in \Cref{subsec:static-analysis}.
    
    \subsubsection{Abstraction Domain and Characteristics}
    
        \label{adxsubsec:abstraction-domain}
    
        \begin{figure}[H]
            \centering
            \includegraphics[width=0.9\linewidth]{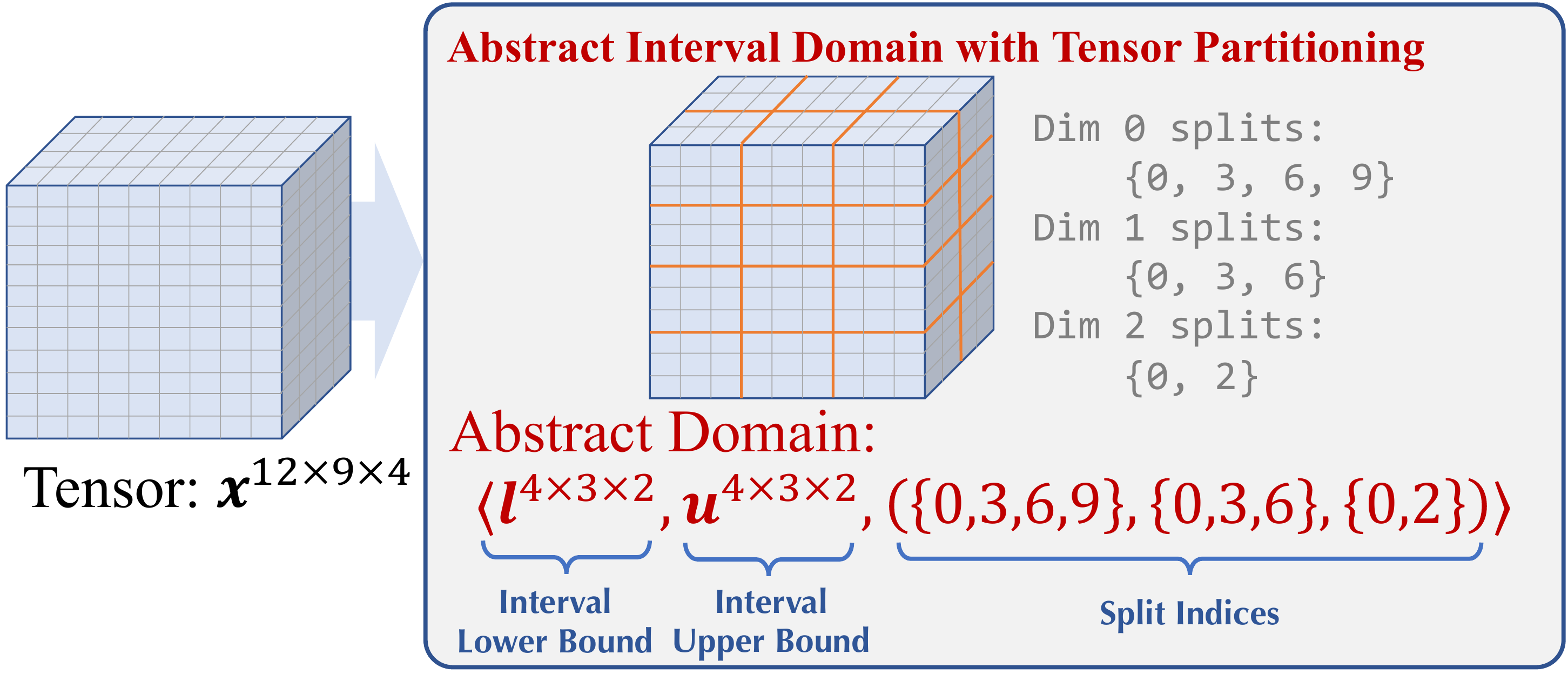}
            \caption{Example of tensor partitioning.
            Tensor partitioning reduces the size of tensors in the abstract domain by sharing one interval bound among all elements inside one subblock. }
            \label{fig:partitioning}
        \end{figure}
    
    We first formally define our abstraction domain: interval with tensor partitioning.
    Following the notation in abstract interpretation literature~\cite{cousot1977static}, suppose the tensor has $m$ dimensions with shape $n_1 \times n_2 \times \cdots \times n_m$, we define the abstract domain as such:
    \begin{equation}
        \begin{aligned}
            & \sA := \{ \langle \vl, \vu, (S_i)_{i=1}^m \rangle: \vl, \vu \in \sR^{n'_1 \times \cdots \times n'_m} \\
            & \hspace{5em} |S_i| = n'_i, \forall s\in S_i, s\in \sN, 0 \le s < n_i \}.
        \end{aligned}
        \label{eq:abs-domain}
    \end{equation}
    Each element in $\sA$ is a triplet $a = \langle \vl, \vu, (S_i)_{i=1}^m\rangle$, where the first two elements are subblock-wise interval lower bound and upper bound, respectively, and the last element $(S_1, S_2, \dots, S_m)$ contains $m$ sets, where each set $S_i$ corresponds to the \emph{0-indexed} split indices for the $i$-th dimension to form subblocks.
    We use $S_i[j] \in \sN$ to represent the $j$-th element of split index set $S_i$ sorted in ascending order and define $S_i[|S_i|] = n_i$.
    \Cref{fig:partitioning} illustrates this abstraction domain.
    
    To define the correspondence between the abstract domain and concrete domain $\sC := 2^{\sR^{n_1\times \cdots \times n_m}}$,
    we form Galois connection $\langle \sC, \subseteq\rangle \underset{\gamma}{\overset{\alpha}{\rightleftharpoons}} \langle \sA, \sqsubseteq \rangle$, where an abstraction function $\alpha: \sC \to \sA$ and the concretization function $\gamma: \sA \to \sC$ are defined as follows:
    \begin{subequations}
    \small
    \begin{equation}
        \begin{aligned}
            & \alpha(\gX) = \langle \vl, \vu, (S_i)_{i=1}^m \rangle \quad \text{where} \\
            & \hspace{2em} \vl_{i_1,i_2,\dots,i_m} = \min_{\vx\in \gX}  \min_{\substack{1\le k\le m\\ S_k[i_k] \le j_k < S_k[i_k + 1]}} \vx_{j_1,j_2,\dots,j_m}, \\
            & \hspace{2em} \vu_{i_1,i_2,\dots,i_m} = \max_{\vx\in \gX} \max_{\substack{1\le k\le m\\ S_k[i_k] \le j_k < S_k[i_k + 1]}} \vx_{j_1,j_2,\dots,j_m}, \\
        \end{aligned}
        \label{eq:alpha}
    \end{equation}
    \vspace{-1em}
    \begin{equation}
        \begin{aligned}
            \gamma(\langle \vl, \vu, (S_i)_{i=1}^m\rangle) = \{ \vx: &\vl_{i_1,i_2,\cdots,i_m} \le \vx_{j_1,j_2,\cdots,j_m} \le \vu_{i_1,i_2,\cdots,i_m}, \\
            & S_k[i_k] \le j_k < S_k[i_k+1], 1\le k\le m \}.
        \end{aligned}
        \label{eq:gamma}
    \end{equation}
    \end{subequations}
    In \Cref{eq:alpha}, the split indices $(S_i)_{i=1}^m$ can be arbitrarily chosen but need to be kept consistent with those in \Cref{eq:gamma}.
    Take \Cref{fig:partitioning} as the example, to abstract a set of tensors $\gX$ with shape $12\times 9 \times 4$, we define split indices for each dimension, respectively, then impose interval constraints on each subblock of the tensor.
    For example, $[\vl_{0,0,0}, \vu_{0,0,0}]$ constrain any element in $\vx_{0:2,0:2,0:1}$, $[\vl_{3,2,1}, \vu_{3,2,1}]$ constrain any element in $\vx_{9:11,6:8,2:3}$. 
    
    \newpage
    
    \textbf{Abstraction Characteristics.}
    There are multiple ways to compute the abstractions, and we design our particular computational algorithms to achieve soundness, tightness, and differentiability.
    \noindent(1)~\textbf{Soundness}: We guarantee that all abstractions are sound.
    Formally, suppose $a_\vx$ and $a_\vw$ are the abstractions for input and weights, respectively, we guarantee $f_n^*(\gamma(a_\vx); \gamma(a_\vw)) \subseteq \gamma(T_n^*(a_\vx; a_\vw))$ where $* \in \{\mathsf{in}, \mathsf{out}\}$, $n$ is any node in the architecture, and $T_n^*(a_\vx; a_\vw)$ is our computed abstract domain for input or output of node $n$.
    The soundness property theoretically guarantees that our detection approach has no false negatives, i.e., flag all potential numerical defects, and the validity of generated preconditions~(see \Cref{subsec:approach-overview}).\\
    (2)~\textbf{Tightness}: For most operators in DNN, given abstractions for its inputs, we compute the tightest possible interval abstraction for the output.
    Formally, for any atomic DNN operator $op$ and any abstract domain $i \in \sA$ of $op$'s input, if $op(\gamma(i)) \subseteq [\vl, \vu]$, i.e., $[\vl, \vu]$ is an interval abstraction of $op$'s output, then $\gamma(T_{op}(i)) \subseteq [\vl, \vu]$, i.e., our generated abstract domain $T_{op}(i)$ is always tighter or equal to any interval abstraction $[\vl, \vu]$.
    Such tightness can reduce false positives for detecting potential numerical defects and increase the span of generated failure-exhibiting intervals and fix intervals, which improves the quality of generated tests and preconditions.\\
    (3)~\textbf{Differentiability}: We compute differentiable abstractions.
    Concretely, if the input of node $n_2$ is deferentially dependent on output of node $n_1$, we can compute out gradients $\nabla_i o$ for any $i \in \{\vl_{n_1}, \vu_{n_1}\}$ and any $o \in \{\vl_{n_2}, \vu_{n_2}\}$, where $\langle \vl_{n_1}, \vu_{n_1}, (S^{n_1}_i)_{i=1}^{\dim(n_1.\mathsf{out})} \rangle$ and $\langle \vl_{n_2}, \vu_{n_2}, (S^{n_2}_i)_{i=1}^{\dim(n_2.\mathsf{in})} \rangle$ are abstractions of $\fout_{n_1}(\cdot; \cdot)$ and $\fin_{n_2}(\cdot; \cdot)$, respectively.
    When tightness and differentiability cannot be achieved at the same time~(e.g., for \CodeIn{floor} operator, tight abstraction $\langle \vl, \vu\rangle \mapsto \langle \lfloor \vl \rfloor, \lfloor \vu \rfloor \rangle$ is not generally differentiable, and differentiable abstraction $\langle \vl, \vu\rangle \mapsto \langle \vl - \mathbf{1}, \vu \rangle$ is not tight), we implement two abstract algorithms to achieve tightness and differentiability, respectively.
    
    The existing approach DEBAR~\cite{zhang2020detecting} also proposes a static analysis framework for DNN architectures with tensor partitioning.
    In contrast to our framework, the abstract domain in DEBAR contains affine equalities besides interval domains.
    Therefore, DEBAR can be tighter than ours in some cases and can produce fewer positives, but due to the additional complexity of affine equalities, DEBAR tends to use the coarsest abstraction~(i.e., tensor partitioning) granularity and supports fewer operators than ours.
    At the same time, our algorithms produce the tightest interval abstractions while DEBAR has no tightness guarantee.
    As a result, \sysname detects more true numerical defects and has a comparable number of false positives~(see \Cref{subsec:static-detector-exp}), and we can leverage the feasibility confirmation support in \sysname to filter out false positives.

    \subsubsection{Initial Abstraction Construction with \textbf{Backward Fine-Grained Node Labeling}}
        \label{adxsubsec:backward-labeling}
        
        We construct abstract domains for initial nodes in two steps:
        First, we determine the tensor partitions, i.e., the split index sets $(S_i)_{i=1}^m$~(see \Cref{eq:abs-domain}), which determine the tightness and efficiency of our static analysis framework, because the tensor partitions of all other nodes will be solely dependent on the partitions of initial nodes as we will show in \Cref{adxsubsec:tight-abstract-dynamic-partitioning}.
        Second, we compute the interval bounds $\vl$ and $\vu$.
        
        We use the following principle to decide the tensor partitions:
        For nodes that are connected to operators requiring fine-grained abstractions~(e.g., the \CodeIn{shape} input for operator \CodeIn{Reshape}) with valid paths~(which we will specify later), we construct tensor partitions with the finest granularity, i.e., $S_i = \{0,1,\dots,n_i\}$.
        We call these nodes fine-grained initial nodes and starting from the next paragraph we introduce our novel technique of \emph{backward fine-grained node labeling} to find them out.
        For other nodes, we rely on downstream task requirements and user specifications to determine the partitions.
        For example, for fix generation~(see \Cref{subsec:precond-gen}), we use the coarsest granularity by default.
        
        \textbf{Backward Fine-Grained Node Labeling.}
        The fine-grained initial nodes are those starting a valid path in DNN computational graph $\gG$, where a path is valid if and only if the path does not traverse through fine-grained stopping operators and terminates at a fine-grained requiring operator with some specific input indices.
        As discussed in \Cref{subsec:static-analysis}, the fine-grained requiring operators fall into three categories: control-flow operators~(e.g., \CodeIn{Loop}), indexing operators~(e.g., \CodeIn{Slice}), and shaping operators~(e.g., \CodeIn{Reshape}). 
        Fine-grained stopping operators are those whose output is independent of input abstraction granularity~(so the granularity of preceding nodes does not matter).
        Detail lists of fine-grained stopping operators and fine-grained requiring operators are provided in \Cref{adxsec:fine-grain-stop-requiring-operators}.
        
        To find out fine-grained initial nodes, we propose backward fine-grained node labeling, which is similar to data dependency analysis in traditional programs:
        First, we invert all edges of the given computational graph $\gG = \langle \gV, \gE \rangle$ to get $\gG' = \langle \gV, \gE' \rangle$.
        Second, we attach a boolean label for each node $n\in \gV$ and initialize them with \CodeIn{False}.
        Third, we do topology sorting on $\gG'$.
        When encountering a node $n$ with \CodeIn{True}, if the node does not contain a fine-grained stopping operator, we propagate this label to all its subsequent nodes; otherwise, we propagate \CodeIn{False}.
        When encountering a node $n$ with \CodeIn{False}, if the node is a fine-grained requiring operator, we propagate \CodeIn{True} to some subsequent nodes corresponding to specific input indices and \CodeIn{False} to others; otherwise, propagate \CodeIn{False} to all subsequent nodes.
        Lastly, we collect all labels for initial nodes. 
        The nodes with \CodeIn{True} label are fine-grained initial nodes.
        
        To this point, we have determined the tensor partitions $(S_i)_{i=1}^m$ for all initial nodes, and we now compute $\vl$ and $\vu$ and thus finish the abstraction domain construction for the valid ranges of initial nodes.
        Initial nodes are further divided into input, weight, and constant nodes.
        In practice, most weight nodes and all constant nodes have their initial values stored in the ONNX file, and we directly use \Cref{eq:alpha} with these initial values to compute $\vl$ and $\vu$.
        Otherwise, we rely on user specifications and built-in heuristics to determine $\vl$ and $\vu$ for initial nodes.
    
    \subsubsection{Internal Abstraction with \textbf{Dynamic Partitioning}}
        \label{adxsubsec:tight-abstract-dynamic-partitioning}
        
        For all supported operator types~(listed in \Cref{adxsec:supported-op}), we propose concrete algorithms to compute abstract domains for output with dynamic tensor partitioning.
        Formally, for each operator $op$, we construct the computable function $T_{op}: \sA \to \sA$ that satisfies soundness and tentatively satisfies tightness and differentiability, where $\sA$ is the abstract domain defined in \Cref{eq:abs-domain}.
        Therefore, following the computational procedure introduced at the end of \Cref{subsec:static-analysis}, we can compute end-to-end abstractions for all nodes in the given DNN architecture.
        As the showcase of our abstraction algorithm, we describe the algorithms for four representative operators: \CodeIn{MatMul}, \CodeIn{Conv}, \CodeIn{Softmax}, and \CodeIn{Loop}.
        
        \paragraph{\CodeIn{MatMul}}
        The \CodeIn{MatMul} operator computes the matrix multiplication of two operands.
        This operator is widely used in DNNs to express fully-connected layers.
        To simplify the narration, we focus on the two-dimensional case where we compute $op(\mA, \mB) = \mC = \mA \mB$ with $\mA \in \sR^{n \times m}$ and $\mB \in \sR^{m \times l}$. 
        Extensions to other dimensions by abstracting the broadcasting mechanism can be found in our open-source implementation.
        
        We denote the input abstractions of $op$ by $a = \langle \mL_a, \mU_a, (S_A^1, S_A^2) \rangle$ and $b = \langle \mL_b, \mU_b, (S_B^1, S_B^2) \rangle$, respectively.
        First, we compute the union $U = S_A^2 \cup S_B^1$.
        Second, we dynamically partition both $a$ and $b$ with split points $(S_A^1, U)$ and $(U, S_B^2)$, respectively, and get
        $a' = \langle \mL'_a, \mU'_a, (S_A^1, U) \rangle$ and $b' = \langle \mL'_b, \mU'_b, (U, S_B^2) \rangle$.
        Note that $a$ and $a'$~(or $b$ and $b'$) correspond to the same concrete domain, but $a'$ and $b'$ have finer or equal partition granularity than $a$ and $b$.
        Third, we compute output abstraction $T_{op}(a,b) = c = \langle \mL_c, \mU_c, (S_A^1, S_B^2) \rangle$
        \begin{equation}
            \small
            \label{eq:mat-abs-1}
            \begin{aligned}
                \text{where} \quad (\mL_c)_{ij} & = \sum_{k=1}^{|U|} \vv_k \min_{\mA \in \{\mL'_a, \mU'_a\}, \mB \in \{\mL'_b, \mU'_b\}} \mA_{ik}\mB_{kj}, \\
                (\mU_c)_{ij} & = \sum_{k=1}^{|U|} \vv_k \max_{\mA \in \{\mL'_a, \mU'_a\}, \mB \in \{\mL'_b, \mU'_b\}} \mA_{ik}\mB_{kj}, \\
                \vv_k & = U[k]-U[k-1].
            \end{aligned}
        \end{equation}
        This formulation can guarantee tightness but is not efficient for tensor computation due to the inner minimum and maximum.
        Therefore, we also implement a fast-mode abstraction trading tightness for efficiency: $T_{op}'(a,b) = c' = \langle \mL_c', \mU_c', (S_A^1, S_B^2)\rangle$ where
        \begin{equation}
            \small
            \begin{aligned}
                \mL_c^{'} = &
                \mU_a^{'-} \vv \mU_b^{'-} + 
                \mU_a^{'0} \vv \mL_b^{'-} + 
                \mU_a^{'+} \vv \mL_b^{'-} + \\
                & \mL_a^{'-} \vv \mU_b^{'0} +
                \mU_a^{'0} \vv \mL_b^{'0} +
                \mL_a^{'0} \vv \mU_b^{'0}
                 +
                \mU_a^{'+} \vv \mL_b^{'0} + \\
                & \mL_a^{'-} \vv \mU_b^{'+} + 
                \mL_a^{'0} \vv \mU_b^{'+} +
                \mL_a^{'+} \vv \mL_b^{'+}, \\
                \mU_c^{'} = & 
                \mL_a^{'-} \vv \mL_b^{'-} +
                \mL_a^{'0} \vv \mL_b^{'-} +
                \mL_a^{'+} \vv \mU_b^{'-} + \\
                & \mL_a^{'-} \vv \mL_b^{'0} +
                \mL_a^{'0} \vv \mL_b^{'0} +
                \mU_a^{'0} \vv \mU_b^{'0}
                 +
                \mU_a^{'+} \vv \mU_b^{'0} + \\
                & \mU_a^{'-} \vv \mL_b^{'+} +
                \mU_a^{'0} \vv \mU_b^{'+} +
                \mU_a^{'+} \vv \mU_b^{'+}.
            \end{aligned}
            \label{eq:mat-abs-2}
        \end{equation}
        In the above equation, for any $* \in \{a, b\}$, let $\circ$ be the elementwise~(Hadamard) product,
        \begin{equation}
            \small
            \begin{aligned}
                \mL_*^{'-} = & \mL_*^{'} \circ \1[\mU_*^{'} < 0],
                \mU_*^{'-} = \mU_*^{'} \circ \1[\mU_*^{'} < 0], \\
                \mL_*^{'0} = & \mL_*^{'} \circ \1[\mL_*^{'} \le 0, \mU_*^{'} \ge 0],
                \mU_*^{'0} = \mU_*^{'} \circ \1[\mL_*^{'} \le 0, \mU_*^{'} \ge 0], \\
                \mL_*^{'+} = & \mL_*^{'} \circ \1[\mL_*^{'} > 0],
                \mU_*^{'+} = \mU_*^{'} \circ \1[\mL_*^{'} > 0].
            \end{aligned}
            \label{eq:mat-abs-2-defs}
        \end{equation}
        From \Cref{eq:mat-abs-2}, we can observe that the abstraction can be easily implemented with tensor computations.
        In \Cref{adxsec:proofs}, we prove the soundness and tightness of these abstractions.
        
        \paragraph{\CodeIn{Conv}}
        The \CodeIn{Conv} operator computes the discrete convolution of two operands.
        This operator is widely used in convolutional neural networks~(CNNs, \cite{krizhevsky2012imagenet}).
        To simplify the narration, we focus on the single-channel single-stride two-dimensional case where we compute $op(\mA, \mW) = \mC$ with $\mA$ being the input matrix and $\mW$ being the convolution kernel. 
        Extensions to general cases are provided in our open-source implementation.
        
        We first dynamically split the kernel abstraction to the finest granularity.
        Second, we compute the receptive field, which is the sub-region of $\mA$ that decides each output position, of each output position.
        Third, we inspect the alignment between receptive fields and $\mA$'s partitions.
        If neighboring output positions have their receptive fields partitioned in the same sub-block of $\mA$, it means that these positions can be abstracted by a single interval, i.e., these positions can be partitioned together.
        Fourth, using this principle, we derive the partitions of the output tensor, and repeat the input tensors accordingly and efficiently so that the abstract computation can be written as convolutional operations.
        Last, we modify the abstraction computation equations in \Cref{eq:mat-abs-2} by replacing matrix multiplication with convolution to compute the abstraction of output $\mC$.
        
        \paragraph{\CodeIn{Softmax}}
        The \CodeIn{Softmax} operator computes normalized exponential values for the given input.
        \CodeIn{Softmax} is widely deployed for classification tasks to yield normalized confidence.
        In one-dimensional case, for input $\vx \in \sR^n$, a \CodeIn{Softmax} operator outputs $op(\vx) = \frac{\exp(\vx)}{\sum_{i=1}^n \exp(\vx)_i}$.
        The output abstraction of \CodeIn{Softmax} operator $op$ can be thus computed: $T_{op}(\langle \vl, \vu, (S)\rangle) = \langle \vl^o, \vu^o, (S)\rangle$ where
        \begin{equation}
            \small
            \resizebox{0.9\linewidth}{!}{$ \displaystyle
            \begin{aligned}
                \vl^o_i & = \dfrac{\exp(\vl_i)}{\exp(\vu)^\T \vv - \exp(\vu_i) + \exp(\vl_i)}, 
                \vu^o_i = \dfrac{\exp(\vu_i)}{\exp(\vl)^\T \vv - \exp(\vl_i) + \exp(\vu_i)}, \\
                \vv_k & = S[k] - S[k-1].
            \end{aligned}
            $}
            \label{eq:softmax-abs}
        \end{equation}
        The output abstraction's partition is dynamically decided by the input abstraction's partition.
        We prove the soundness and tightness of the abstraction in \Cref{adxsec:proofs}.

        \paragraph{\CodeIn{Loop}}
        The node with a \CodeIn{Loop} operator contains a sub-computational graph with dynamic controlling counters and conditions to represent a runtime loop.
        The \CodeIn{Loop} operator can be used to represent recurrent neural networks~(RNNs) and handle the input sequence of variable length.
        We require the controlling counter and loop termination condition to have the finest abstraction granularity.
        Then, we recursively apply our static analysis framework to the sub-graph representing the loop body and update the interval of the loop counter iteratively.
        When the abstraction interval of the loop counter can explicitly decide whether to terminate the loop, we continue or terminate the loop iterations, respectively;
        otherwise, we merge the current abstraction interval with the interval obtained after another iteration.
        We repeat this process until termination.
        Theoretically, this execution process cannot guarantee the termination, i.e., the loop body may execute for infinite times.
        However, in practice, on our dataset, our analysis framework already suffices to guarantee the termination in all cases.
        In the future, we can apply a widening operation if termination cannot be tightly abstracted. 
        
        \begin{remark}
            As we can see, the novel dynamic partition technique is incorporated into the computation process of each operator's abstraction.
            The soundness and tightness of our designed abstractions are immediately achieved by design, and the differentiability of our designed abstractions can be implemented via the auto-differentiation functionality of  popular DL libraries like \CodeIn{PyTorch}~\cite{paszke2019pytorch} and \CodeIn{Tensorflow}~\cite{tensorflow2015-whitepaper} where we use \CodeIn{PyTorch} for implementation. 
            The abstractions for DNNs are also implemented for other applications, e.g., for robustness verification~\cite{singh2019abstract,mirman2018differentiable}.
            However, our abstractions are tailored for the tensor-partitioned interval domain that is particularly suitable for testing and debugging the numerical defects.
            To the best of our knowledge, these abstractions are the first that achieve soundness, tightness, and differentiability.
        \end{remark}
    
    \subsubsection{List of Fine-Grained Requiring and Stopping Operators}
    
        \label{adxsec:fine-grain-stop-requiring-operators}
        
        We use fine-grained requiring and fine-grained stopping operators in \Cref{adxsubsec:backward-labeling}.
        Among all supported operators, the fine-grained requiring operators are 
        \vspace{0.5em}
        
        \begin{tabular}{p{8.1cm}}
            \CodeIn{Reshape} (input \#2), 
            \CodeIn{Slice} (input \#2, \#3, \#4, \#5), 
            \CodeIn{Squeeze} (input \#2), 
            \CodeIn{Unsqueeze} (input \#2), 
            \CodeIn{Tile} (input \#2, \#3), 
            \CodeIn{Loop} (input \#1, \#2),
            \CodeIn{SequenceInsert} (input \#3),
            \CodeIn{ConstantOfShape} (input \#1),
            \CodeIn{Gather} (input \#2),
            \CodeIn{GatherND} (input \#2),
            \CodeIn{ReduceSum} (input \#2),
            \CodeIn{ScatterElements} (input \#2),
            \CodeIn{Expand} (input \#2),
            \CodeIn{Split} (input \#2),
            \CodeIn{Pad} (input \#2, \#3),
            \CodeIn{NegativeLogLikelihoodLoss} (input \#2),
            \CodeIn{Clip} (input \#2, \#3),
            \CodeIn{OneHot} (input \#2),
            \CodeIn{Resize} (input \#2, \#3, \#4).
        \end{tabular}
        \vspace{0.5em}

        \noindent
        The fine-grained stopping operators are 
        \vspace{0.5em}
        
        \begin{tabular}{p{8.1cm}}
            \CodeIn{Shape}, \CodeIn{RandomNormalLike}, \CodeIn{RandomUniformLike}.
        \end{tabular}
        
        \vspace{0.5em}
            
    \subsubsection{Proofs}
        \label{adxsec:proofs}
        \allowdisplaybreaks
        
        Here we present the omitted proofs in \Cref{adxsubsec:tight-abstract-dynamic-partitioning}.
        
        \paragraph{\CodeIn{MatMul}}
        
        \begin{theorem}[Tightness of Abstraction by \Cref{eq:mat-abs-1} for \CodeIn{MatMul}]
            Suppose $op$ is the \CodeIn{MatMul} operator and $T_{op}$ is as defined in \Cref{eq:mat-abs-1}, then if $op(\gamma(a), \gamma(b)) \subseteq [\mL, \mU]$ for $a, b \in \sA$, $\gamma(T_{op}(a, b)) \subseteq [\mL, \mU]$.
            \label{thm:1}
        \end{theorem}
        
        \begin{proof}
            We use $\mL_a$ and $\mU_a$ to denote the element-wise interval lower and upper bounds of $\gamma(a)$; and use $\mL_b$ and $\mU_b$ to denote these bounds of $\gamma(b)$, respectively, where a formal definition is in \Cref{eq:gamma}.
            Then, there exist $\mA$ and $\mB$ with $\mL_a \le \mA \le \mU_a$ and $\mL_b \le \mB \le \mU_b$, such that
            \begin{small}
                \begin{align}
                    (\mA\mB)_{ij} & = \sum_{k=1}^l \min \{ \mL_{a,ik} \mL_{b,kj}, \mL_{a,ik} \mU_{b,kj}, \nonumber \\
                    & \hspace{5em} \mU_{a,ik} \mL_{b,kj}, \mU_{a,ik} \mU_{b,kj} \}, \label{eq:pf-1} \\
                    \text{or } (\mA\mB)_{ij} & = \sum_{k=1}^l \max \{ \mL_{a,ik} \mL_{b,kj}, \mL_{a,ik} \mU_{b,kj}, \nonumber \\
                    & \hspace{5em} \mU_{a,ik} \mL_{b,kj}, \mU_{a,ik} \mU_{b,kj} \}. \label{eq:pf-2} 
                \end{align}
            \end{small}
            By definition, we have
            \begin{equation}
                \small
                \mL_{ij} \le \text{\Cref{eq:pf-1}}, \mU_{ij} \ge \text{\Cref{eq:pf-2}}.
            \end{equation}
            We let $\mL'_c$ and $\mU'_c$ to denote the element-wise interval lower and upper bounds for $\gamma(T_{op}(a,b))$, 
            let $S_A^1[i'] \le i-1 < S_A^1[i'+1]$ and $S_B^2[j'] \le j-1 < S_B^2[j'+1]$,
            then from \Cref{eq:mat-abs-1},
            \begin{equation}
                \small
                \begin{aligned}
                & (\mL'_c)_{ij} = (\mL_c)_{i'j'} \\
                = & \sum_{k=1}^{|U|} \sum_{k'=U[k-1]+1}^{U[k]} \min \{ \mL_{a,ik'} \mL_{b,k'j}, \mL_{a,ik'} \mU_{b,k'j}, \\
                & \hspace{10em} \mU_{a,ik'} \mL_{b,k'j}, \mU_{a,ik'} \mU_{b,k'j} \} \\
                = & \sum_{k=1}^l \min \{ \mL_{a,ik} \mL_{b,kj}, \mL_{a,ik} \mU_{b,kj}, \mU_{a,ik} \mL_{b,kj}, \mU_{a,ik} \mU_{b,kj} \} \\
                = & \text{\Cref{eq:pf-1}} \ge \mL_{ij}.
                \end{aligned}
                \label{eq:pf-5}
            \end{equation}
            Similarly, $(\mU'_c)_{ij} \le \mU_{ij}$.
            Thus, $\gamma(T_{op}(a,b)) \subseteq [\mL, \mU]$.
        \end{proof}
        
        \begin{theorem}[Soundness of Abstraction by \Cref{eq:mat-abs-1} for \CodeIn{MatMul}]
            Suppose $op$ is the \CodeIn{MatMul} opeartor and $T_{op}$ is as defined in \Cref{eq:mat-abs-1}, then $op(\gamma(a), \gamma(b)) \subseteq \gamma(T_{op}(a,b))$.
            \label{thm:2}
        \end{theorem}
        
        \begin{proof}
            We use $\mL_a$ and $\mU_a$ to denote the element-wise interval lower and upper bounds of $\gamma(a)$; and use $\mL_b$ and $\mU_b$ to denote these bounds of $\gamma(b)$, respectively, where a formal definition is in \Cref{eq:gamma}.
            For any $\mA$ and $\mB$ such that $\mL_a \le \mA \le \mU_a$ and $\mL_b \le \mB \le \mU_b$, 
            \begin{equation}
                \small
                \begin{aligned}
                    (\mA\mB)_{ij} & = \sum_{k=1}^l \mA_{ik} \mB_{kj} \\
                    & \ge \min\{ \mL_{a,ik} \mL_{b,kj}, \mL_{a,ik} \mU_{b,kj}, \mU_{a,ik} \mL_{b,kj}, \mU_{a,ik} \mU_{b,kj} \} \\
                    & =: (\mL_{ab}')_{ij}
                \end{aligned}
                \label{eq:pf-3} 
            \end{equation}
            and
            \begin{equation}
                \small
                \begin{aligned}
                    (\mA\mB)_{ij} & \le \max\{ \mL_{a,ik} \mL_{b,kj}, \mL_{a,ik} \mU_{b,kj}, \mU_{a,ik} \mL_{b,kj}, \mU_{a,ik} \mU_{b,kj} \} \\
                    & =: (\mU_{ab}')_{ij}.
                \end{aligned}
                \label{eq:pf-4} 
            \end{equation}
            Thus, $op(\gamma(a), \gamma(b)) \subseteq [\mL'_{ab}, \mU'_{ab}]$.
            On the other hand, $\gamma(T_{op}(a,b)) = [\mL'_{ab}, \mU'_{ab}]$ as seen from \Cref{eq:pf-5}.
            Therefore, $op(\gamma(a), \gamma(b)) \subseteq \gamma(T_{op}(a,b))$.
        \end{proof}
    
        \begin{theorem}[Soundness of Abstraction by \Cref{eq:mat-abs-2} for \CodeIn{MatMul}]
            Suppose $op$ is the \CodeIn{MatMul} operator and $T'_{op}$ is as defined in \Cref{eq:mat-abs-2,eq:mat-abs-2-defs}, then $op(\gamma(a), \gamma(b)) \subseteq \gamma(T'_{op}(a,b))$.
            \label{thm:3}
        \end{theorem}
        
        \begin{proof}
            We use $\mL_a$ and $\mU_a$ to denote the element-wise interval lower and upper bounds of $\gamma(a)$; and use $\mL_b$ and $\mU_b$ to denote these bounds of $\gamma(b)$, respectively, where a formal definition is in \Cref{eq:gamma}.
            We define $\mL'_{ab}$ and $\mU'_{ab}$ by \Cref{eq:pf-3,eq:pf-4}.
            From the proof of \Cref{thm:2}, we have $op(\gamma(a), \gamma(b)) \subseteq [\mL'_{ab}, \mU'_{ab}]$.
            We now only need to show that $[\mL'_{ab}, \mU'_{ab}] \subseteq \gamma(T'_{op}(a,b))$.
            
            $\gamma(T'_{op}(a,b))$ imposes independent interval abstractions element-wise, therefore, we study each element independently.
            For the element $(i,j)$, from \Cref{eq:mat-abs-2,eq:mat-abs-2-defs}, the interval lower bound of $\gamma(T'_{op}(a,b))$, namely $(\mL_c)_{ij}$, satisfies
            \begin{equation}
                \small
                \begin{aligned}
                    & (\mL_c)_{ij} =  \sum_{k=1}^l  
                    (\mU_a^-)_{ik} (\mU_b^-)_{kj} +
                    (\mU_a^{0})_{ik} (\mL_b^{-})_{kj} + 
                    (\mU_a^{+})_{ik} (\mL_b^{-})_{kj} + \\
                    & (\mL_a^{-})_{ik} (\mU_b^{0})_{kj} +
                    (\mU_a^{0})_{ik} (\mL_b^{0})_{kj} +
                    (\mL_a^{0})_{ik} (\mU_b^{0})_{kj} +
                    (\mU_a^{+})_{ik} (\mL_b^{0})_{kj} + \\
                    & (\mL_a^{-})_{ik} (\mU_b^{+})_{kj} + 
                    (\mL_a^{0})_{ik} (\mU_b^{+})_{kj} +
                    (\mL_a^{+})_{ik} (\mL_b^{+})_{kj}.
                \end{aligned}
                \label{eq:pf-6}
            \end{equation}
            By \Cref{eq:mat-abs-2-defs},
            \begin{itemize}
                \item when $\mL_{a,ik} \le \mU_{a,ik} < 0$, \\
                $(\mL_a^-)_{ik} = \mL_{a,ik}, (\mU_a^-)_{ik} = \mU_{a,ik}, 
                (\mL_a^0)_{ik} = 0, \\ (\mU_a^0)_{ik} = 0, 
                (\mL_a^+)_{ik} = 0, (\mU_a^+)_{ik} = 0$;
                
                \item when $\mL_{a,ik} \le 0 \le \mU_{a,ik}$, \\
                $(\mL_a^-)_{ik} = 0, (\mU_a^-)_{ik} = 0, 
                (\mL_a^0)_{ik} = \mL_{a,ik}, \\ (\mU_a^0)_{ik} = \mU_{a,ik}, 
                (\mL_a^+)_{ik} = 0, (\mU_a^+)_{ik} = 0$;
                
                \item when $0 < \mL_{a,ik} \le \mU_{a,ik}$, \\
                $(\mL_a^-)_{ik} = 0, (\mU_a^-)_{ik} = 0, 
                (\mL_a^0)_{ik} = 0, \\ (\mU_a^0)_{ik} = 0, 
                (\mL_a^+)_{ik} = \mL_{a,ik}, (\mU_a^+)_{ik} = \mU_{a,ik}$.
            \end{itemize}
            Similarly for $(\mL_b^-)_{kj}$, $(\mU_b^-)_{kj}$, $(\mL_b^0)_{kj}$, $(\mU_b^0)_{kj}$, $(\mL_b^+)_{kj}$ and $(\mU_b^+)_{kj}$.
            Thus, by enumerating all cases, we have
            \begin{equation}
                \small
                \begin{aligned}
                \text{\Cref{eq:pf-6}} \le \min\{ & \mL_{a,ik} \mL_{b,kj},\mL_{a,ik} \mU_{b,kj},  \\
                & \mU_{a,ik} \mL_{b,kj}, \mU_{a,ik} \mU_{b,kj} \} = (\mL'_{ab})_{ij}.
                \end{aligned}
            \end{equation}
            Similarly, the interval lower bound of $\gamma(T'_{op}(a,b))$, namely $(\mU_c)_{ij}$,
            \begin{equation}
                \small
                \begin{aligned}
                (\mU_c)_{ij} \ge \max\{ & \mL_{a,ik} \mL_{b,kj}, \mL_{a,ik} \mU_{b,kj}, \\
                & \mU_{a,ik} \mL_{b,kj}, \mU_{a,ik} \mU_{b,kj} \} = (\mU'_{ab})_{ij}.
                \end{aligned}
            \end{equation}
            Thus, $[\mL'_{ab}, \mU'_{ab}] \subseteq \gamma(T'_{op}(a,b))$.
        \end{proof}
        
        \paragraph{\CodeIn{Softmax}}
        
        \begin{theorem}[Tightness of Abstraction for \CodeIn{Softmax}]
            Suppose $op: \sR^n \to \sR^n$ is the \CodeIn{Softmax} operator and $T_{op}$ is as defined in \Cref{eq:softmax-abs}, if $op(\gamma(\langle \vl, \vu, (S)\rangle)) \subseteq [\vl^r, \vu^r]$, then $\gamma(T_{op}(\langle \vl, \vu, (S)\rangle)) \subseteq [\vl^r, \vu^r]$.
            \label{thm:4}
        \end{theorem}
        
        \begin{proof}
            We use $\vl'$ and $\vu'$ to denote the element-wise interval lower and upper bounds of $\gamma(\langle \vl, \vu, (S)\rangle)$.
            Formally, $\vl'_i = \vl_{i'}$ and $\vu'_i = \vu_{i'}$ where $S[i'] \le i - 1 < S[i'+1]$.
            Then, 
            for each $i$, by setting $\vx_i = \vl_i$ and $\vx_j = \vu_j$ for all $j \neq i$, we get
            \begin{equation}
                \small
                \begin{aligned}
                    \vl_i^r & \le \dfrac{\exp(\vl'_i)}{\exp(\vl'_i) + \sum_{k=1, k\neq i}^n \exp(\vu'_k)} \\
                    & = \dfrac{\exp(\vl_{i'})}{\exp(\vl_{i'}) + \sum_{k=1}^n \exp(\vu'_k) - \exp(\vu'_i)} \\
                    & = \dfrac{\exp(\vl_{i'})}{\exp(\vl_{i'}) + \vv^\T \exp(\vu_k) - \exp(\vu_{i'})} = \vl_{i'}^o.
                \end{aligned}
            \end{equation}
            In addition, by setting $\vx_i = \vu_i$ and $\vx_j = \vl_j$ for all $j \neq i$, we get $\vu_i^r \ge \vu_{i'}^o$.
            Here, $\vl^o$ and $\vu^o$ are as defined in \Cref{eq:softmax-abs}.
            Combining these two arguments, we get 
            $$\gamma(T_{op}(\langle \vl, \vu, (S)\rangle)) \subseteq [\vl^r, \vu^r].$$
        \end{proof}
    
        \begin{theorem}[Soundness of Abstraction for \CodeIn{Softmax}]
            Suppose $op: \sR^n \to \sR^n$ is the \CodeIn{Softmax} operator and $T_{op}$ is as defined in \Cref{eq:softmax-abs}, then $op(\gamma(\langle \vl, \vu, (S)\rangle)) \subseteq \gamma(T_{op}(\langle \vl, \vu, (S)\rangle))$.
            \label{thm:5}
        \end{theorem}
        
        \begin{proof}
            Leveraging the fact that the softmax function is monotonically increasing, i.e., $\frac{\d op(\vx)_i}{\d \vx_j} > 0$, we have
            $\gamma(\langle \vl, \vu, (S)\rangle)_i \in [\vl^o_{i'}, \vu^o_{i'}]$,
            where $S[i'] \le i-1 < S[i'+1]$.
            Since $[\vl^o_{i'}, \vu^o_{i'}] = \gamma(T_{op}(\langle \vl, \vu, (S)\rangle))_i$ by our definition of $T_{op}$, $op(\gamma(\langle \vl, \vu, (S)\rangle)) \subseteq \gamma(T_{op}(\langle \vl, \vu, (S)\rangle))$ follows.
        \end{proof}

\balance

\subsection{Implementation}

    \label{sec:impl}
    
    We have implemented a tool for \sysname in roughly 10k lines of Python code based on PyTorch.
    Our tool leverages existing modules~(\CodeIn{tf2onnx} for PyTorch and \CodeIn{torch.onnx.export} for Tensorflow and Keras) to extract ONNX-format DNN architectures from DL programs to take as the input, and automatically generates detection results, system tests, and defect fixes~(given imposing locations).
    

\subsection{Hyperparameters}

    \label{adxsec:detail-hyperparam}

    In this section, we listed hyperparameters in our two-step system test generation technique:
    %
    For the unit test generation, we use the Adam optimizer where the learning rate is $1$ and the maximum iteration number is $100$.
    For the training example generation, we target for training example under learning rate $\gamma=1$ and the approach has similar performance under other learning rates.
    We follow the convention in the DLG attack, where we use the L-BFGS method as the optimizer for gradient-based minimization.
    We terminate the method and return ``\CodeIn{failed}'' if either the running time exceeds \SI{1800}{s}~(universal execution time limit for all experimental evaluations), or a failure-exhibiting example training is not found after $300$ iterations of L-BFGS optimization.
    
    
    
\subsection{Ablation Study of System Test Generation}

    \label{adxsec:system-testgen-by-stage}
    
    \sysname uses the two-step test generation technique to
    produce failure-exhibiting system tests:
    it first generates failure-exhibiting unit tests with gradient back-propagation, then generates failure-exhibiting training examples via the extended DLG attack.
    To isolate the impact of \sysname at each step, we replace either step with random sampling:
    ``Random + \sysname'', which first generates failure-exhibiting unit tests via random sampling, then generates training example via \sysname;
    ``\sysname + Random'', which first generates failure-exhibiting unit tests via \sysname, then generates training example via random sampling.
    We follow the same evaluation protocol as in \Cref{subsec:system-test-gen-exp} in this ablation study.
    We find that ``\sysname + Random'' takes 9.38X running time than \sysname and fails for 68 runs~(\sysname only fails for 57 runs); and ``Random + \sysname'' fails for 113 runs~(roughly 2X failed runs compared to \sysname).
    This study implies that \sysname's technique helps to improve effectiveness and efficiency at both steps of failure-exhibiting system test generation compared to the pure random baseline.
    The improvement for the first step is mainly from the effectiveness perspective, and the improvement for the second step is mainly from the efficiency perspective.
    




\subsection{Threats to Validity}

    An external threat to validity is the evaluation subjects, which are the GRIST benchmarks~\cite{yan2021exposing}, used to  evaluate \sysname and baseline approaches. 
    Although the subjects  contain 63 real-world DL programs and are the largest to our knowledge, they still may not be representative enough.
    Another external threat is the approach randomness.
    To mitigate this threat, we repeat all randomized approaches for 10 runs.
    To mitigate bias in the empirical study, two authors independently conduct the study and discuss all cases to reach a consensus. A major internal threat to validity comes from the approach implementation.
    We reduce this threat by conducting code review  and log checking  extensively.
    An author independently checks the correctness of code implementation.
    We also verify the soundness of our static analysis framework with over 50 carefully designed unit tests.

\end{document}